\pdfoutput=1

\newif\iflong
\longtrue{} %

\iflong 
\documentclass[acmsmall,screen,noacm]{acmart}
\else
\documentclass[acmsmall,screen]{acmart}
\fi

\setcopyright{none}
\setcctype{by}
\acmJournal{PACMPL}
\acmYear{2025} \acmVolume{9} \acmNumber{PLDI} \acmArticle{174} \acmMonth{6} \acmPrice{}\acmDOI{10.1145/3729277}
\usepackage{blkarray}
\usepackage{xparse}
\usepackage{booktabs}
\usepackage{siunitx}
\usepackage{microtype}
\usepackage{amsmath}
\usepackage{amsfonts}
\usepackage{stmaryrd}
\usepackage{thm-restate}
\usepackage{multicol}
\usepackage{color}
\usepackage{algorithm}
\usepackage{float}
\usepackage{multirow}
\usepackage[noend]{algpseudocode}
\usepackage{graphicx}
\usepackage{tikz}
\usepackage{thm-restate}
\usepackage{mathpartir}
\usepackage{xspace}
\usepackage{paralist}
\usepackage[inline]{enumitem}
\usepackage[capitalise,nameinlink]{cleveref}
\usetikzlibrary{calc,tikzmark}
\usetikzlibrary{arrows,automata}
\usetikzlibrary{positioning}
\usepackage{ifthen}
\usepackage{xargs}
\usepackage{wrapfig}
\usepackage{csvsimple}
\usepackage{siunitx}
\usepackage{fp}
\usepackage{pifont}
\usepackage{xcolor}

\usepackage{tabularray}
\UseTblrLibrary{booktabs}

\usepackage{balance}
\usepackage{etoolbox}
\usepackage{adjustbox}
\usepackage{pgfplots}
\pgfplotsset{compat=1.18}
\usepackage{subcaption}
\usepackage{bookmark}
\usepackage{booktabs}
\usepackage{accents}
\usepackage{datatool}
\usepackage{catchfile}

\usepackage{mdframed}

\newcommand{\cmark}{{\color{green!70!black}\ding{51}}}%
\newcommand{\xmark}{{\color{red!70!black}\ding{55}}}%
\newcommand{\checkmarkt}[1]{%
	\edef\TVALUE{{#1}}%
	\expandafter\ifstrequal\TVALUE{yes}{\cmark}{}%
	\expandafter\ifstrequal\TVALUE{no}{\xmark}{}%
			\expandafter\ifstrequal\TVALUE{no*}{\xmark $^\star$}{}%
}

\theoremstyle{acmdefinition}

\AtEndPreamble{%
\theoremstyle{acmplain}

\newtheorem*{notation*}{Notation}
}

\newcommand{\theoremhang}{%
  \begingroup%
    \setlength{\unitlength}{.005\linewidth}%
    \begin{picture}(0,0)(0,0)%
      \linethickness{0.45pt} \color{black!50}%
      \put(-1,2){\line(1,0){201.5}}%
      \put(-1,2){\line(0,-1){4}}%
      \put(200.5,2){\line(0,-1){4}}%
    \end{picture}%
  \endgroup%
}%

\newcommand{\theoremhung}{%
  \begingroup%
    \setlength{\unitlength}{.005\linewidth}%
    \begin{picture}(0,0)(0,0)%
      \linethickness{0.45pt} \color{black!50}%
      \put(-1,0){\line(1,0){201.5}}%
      \put(-1,0){\line(0,1){4}}%
      \put(200.5,0){\line(0,1){4}}%
    \end{picture}%
  \endgroup%
}%

\newcounter{example}[section]
\newenvironment{example}
  {\par\noindent\theoremhang\par\nobreak\vspace{-6pt}\noindent
   \refstepcounter{example}\postdisplaypenalty=10000 %
   {\indent\scshape Example \theexample.}\ \ignorespaces}
  {\par\nobreak\vspace{-6pt}\noindent\theoremhung\par\addvspace{\topsep}}
\numberwithin{example}{section}

\crefformat{section}{#2\S{}#1#3}
\Crefname{section}{Section}{Section}
\Crefformat{section}{Section #2#1#3}
\crefname{corollary}{\text{Corollary}}{\text{corollaries}}
\Crefname{corollary}{\text{Corollary}}{\text{Corollaries}}
\crefname{lemma}{\text{Lemma}}{\text{Lemmas}}
\Crefname{lemma}{\text{Lemma}}{\text{Lemmas}}
\crefname{proposition}{\text{Prop.}}{\text{Propositions}}
\Crefname{proposition}{\text{Proposition}}{\text{Propositions}}
\crefname{definition}{\text{Def.}}{\text{Definitions}}
\Crefname{definition}{\text{Definition}}{\text{Definitions}}
\crefname{notation}{\text{Notation}}{\text{Notations}}
\Crefname{notation}{\text{Notation}}{\text{Notations}}
\crefname{theorem}{\text{Thm.}}{\text{Theorems}}
\Crefname{theorem}{\text{Theorem}}{\text{Theorems}}
\crefname{figure}{\text{Fig.}}{\text{Figures}}
\Crefname{figure}{\text{Figure}}{\text{Figures}}
\crefname{example}{\text{Ex.}}{\text{Examples}}
\Crefname{example}{\text{Example}}{\text{Examples}}
\crefformat{appendix}{#2\S{}#1#3}

\newcommand{\noop}[1]{}



\AtEndPreamble{%
\newcounter{claimcounter}
\numberwithin{claimcounter}{theorem}

\crefname{claimcounter}{\text{Claim}}{\text{Claims}}
\Crefname{claimcounter}{\text{Claim}}{\text{Claims}}
}

\newcommand{\textcode}[1]{\texorpdfstring{\texttt{#1}}{#1}}
\newcommand{\kw}[1]{\textbf{\textcode{#1}}}
\newcommand{\skipc}{\kw{skip}}

\newcommand{\while}[2]{\kw{while}\;#1\;\kw{do}\\ \quad\;#2}

\newcommand{\ie}{{i.e.,} }
\newcommand{\eg}{{e.g.,} }

\newcommand{\aka}{a.k.a.~}

\newcommand{\inarr}[1]{\begin{array}{@{}l@{}}#1\end{array}}
\newcommand{\inarrII}[2]{\begin{array}{@{}l@{~~}||@{~~}l@{}}\inarr{#1}&\inarr{#2}\end{array}}

\newcommand{\set}[1]{\{{#1}\}}

\newcommand{\st}{\; | \;}
\newcommand{\N}{{\mathbb{N}}}
\newcommand{\dom}[1]{\textit{dom}{({#1})}}

\newcommand{\tup}[1]{{\langle{#1}\rangle}}
\newcommand{\nin}{\not\in}
\newcommand{\suq}{\subseteq}

\newcommand{\size}[1]{|{#1}|}

\newcommand{\defeq}{\triangleq}

\makeatletter
\newcommand{\raisemath}[1]{\mathpalette{\raisem@th{#1}}}
\newcommand{\raisem@th}[3]{\raisebox{#1}{$#2#3$}}
\makeatother

\newcommandx{\yaHelper}[2][1=\empty]{%
\ifthenelse{\equal{#1}{\empty}}%
  { \ensuremath{ \scriptstyle{ #2 } } } %
  { \raisebox{ #1 }[0pt][0pt]{ \ensuremath{ \scriptstyle{ #2 } } } }  %
}

\newcommandx{\yrightarrow}[4][1=\empty, 2=\empty, 4=\empty, usedefault=@]{%
  \ifthenelse{\equal{#2}{\empty}}
  { \xrightarrow{ \protect{ \yaHelper[ #4 ]{ #3 } } } } %
  { \xrightarrow[ \protect{ \yaHelper[ #2 ]{ #1 } } ]{ \protect{ \yaHelper[ #4 ]{ #3 } } } } %
}

\colorlet{colorPO}{gray!60!black}
\colorlet{colorRF}{green!60!black}
\colorlet{colorMO}{orange}
\colorlet{colorTS}{orange}
\colorlet{colorFR}{purple}
\colorlet{colorECO}{red!80!black}
\colorlet{colorSYN}{green!40!black}
\colorlet{colorHB}{blue}
\colorlet{colorPPO}{magenta}
\colorlet{colorPB}{olive}
\colorlet{colorSBRF}{olive}
\colorlet{colorRMW}{olive!70!black}
\colorlet{colorRSEQ}{blue}
\colorlet{colorSC}{violet}
\colorlet{colorPSC}{violet}
\colorlet{colorREL}{olive}
\colorlet{colorCONFLICT}{olive}
\colorlet{colorRACE}{olive}
\colorlet{colorWB}{orange!70!black}
\colorlet{colorPSC}{violet}
\colorlet{colorSCB}{violet}
\colorlet{colorDEPS}{violet}

\tikzset{
   every path/.style={>=stealth},
   po/.style={->,color=colorPO,thin,shorten >=-0.5mm,shorten <=-0.5mm},
   sw/.style={->,color=colorSYN,shorten >=-0.5mm,shorten <=-0.5mm},
   rf/.style={->,color=colorRF,dashed,thick,shorten >=-0.5mm,shorten <=-0.5mm},
   hb/.style={->,color=colorHB,thick,shorten >=-0.5mm,shorten <=-0.5mm},
   mo/.style={->,color=colorMO,dotted,very thick,shorten >=-0.5mm,shorten <=-0.5mm},
   no/.style={->,dotted,thick,shorten >=-0.5mm,shorten <=-0.5mm},
   fr/.style={->,color=colorFR,dotted,thick,shorten >=-0.5mm,shorten <=-0.5mm},
   deps/.style={->,color=colorDEPS,dotted,thick,shorten >=-0.5mm,shorten <=-0.5mm},
   rmw/.style={->,color=colorRMW,thick,shorten >=-0.5mm,shorten <=-0.5mm},
   revisit/.style={rounded corners,fill=yellow!50!gray},
   revisit2/.style={rounded corners,fill=yellow!50!gray!50!white},
}

\newcommand{\na}{\mathtt{na}}

\newcommand{\rlx}{\mathtt{rlx}}
\newcommand{\rel}{{\mathtt{rel}}}
\newcommand{\acq}{{\mathtt{acq}}}
\newcommand{\acqrel}{{\mathtt{acqrel}}}
\newcommand{\sco}{{\mathtt{sc}}}

\newcommand{\evlab}[4]{{#1}^{#2}({#3},{#4})}
\newcommand{\rlab}[3]{{\lR}^{#1}({#2},{#3})}
\newcommand{\wlab}[3]{{\lW}^{#1}({#2},{#3})}

\newcommand{\ulab}[4]{{\lU}^{#1}({#2},{#3},{#4})}

\newcommand{\lE}{{\mathtt{E}}}

\newcommand{\lR}{{\mathtt{R}}}

\newcommand{\lW}{{\mathtt{W}}}

\newcommand{\lU}{{\mathtt{RMW}}}

\newcommand{\sE}{\mathsf{E}}
\newcommand{\sR}{\mathsf{R}}
\newcommand{\sW}{\mathsf{W}}
\newcommand{\sU}{\mathsf{RMW}}

\newcommand{\sF}{\mathsf{F}}

\newcommand{\lTID}{{\mathtt{tid}}}

\newcommand{\lLOC}{{\mathtt{loc}}}
\newcommand{\lMOD}{{\mathtt{mod}}}
\newcommand{\lVAL}{{\mathtt{val}}}

\newcommand{\po}{{\color{colorPO}\mathit{po}}}
\newcommand{\rf}{{\color{colorRF}\mathit{rf}}}

\newcommand{\mo}{{\color{colorMO}\mathit{mo}}}

\newcommand{\seq}{\mathbin{;}}

\newcommand{\lPO}{{\color{colorPO}\mathtt{po}}}
\newcommand{\lRF}{{\color{colorRF} \mathtt{rf}}}

\newcommand{\lMO}{{\color{colorMO} \mathtt{mo}}}

\newcommand{\lFR}{{\color{colorFR} \mathtt{fr}}}

\newcommand{\lSW}{{\color{colorSYN}\mathtt{sw}}}
\newcommand{\lHB}{{\color{colorHB}\mathtt{hb}}}
\newcommand{\lHBSC}{{\color{colorHB}\mathtt{hb}}_{\SC}}

\newcommand{\Init}{\mathsf{Init}}
\newcommand{\Tid}{\mathsf{Tid}}
\newcommand{\Loc}{\mathsf{Loc}}
\newcommand{\NALoc}{\mathsf{Loc_{\na}}}
\newcommand{\Val}{\mathsf{Val}}

\newcommand{\Lab}{\mathsf{Lab}}

\newcommand{\Crit}{\mathsf{C}}

\newcommand{\fenceInst}[1]{\kw{fence}({#1})}

\newcommand{\assignInst}[2]{#1\;{:=}\;#2}
\newcommand{\faddInst}[3]{#1 \;{:=}\; \faddInstn({#2},{#3})}

\newcommand{\casInstn}{\kw{CAS}}
\newcommand{\faddInstn}{\kw{FADD}}
\newcommand{\waitInstn}{\kw{wait}}
\newcommand{\bcasInstn}{\kw{BCAS}}
\newcommand{\waitInst}[2]{{\waitInstn({#1} = {#2})}}

\newcommand{\RLX}{\ensuremath{\textnormal{RC20}}\xspace}
\newcommand{\SC}{\ensuremath{\textnormal{SC}}\xspace}

\newcommand{\HB}{\ensuremath{\textnormal{HB}}\xspace}

\newcommand{\prev}{\mathit{prev}}

\newcommand{\vc}{\mathbb{L}}

\newcommand{\vcT}{\mathbb{{}{T}}}
\newcommand{\vcW}{\mathbb{{}{W}}}
\newcommand{\vcM}{\mathbb{{}{M}}}

\newcommand{\bmT}{\mathrm{T}}
\newcommand{\bmW}{\mathrm{W}}
\newcommand{\bmM}{\mathrm{M}}
\newcommand{\vcU}{\lU}
\newcommand{\vcTU}{{\vcT}^{\vcU}}
\newcommand{\vcWU}{{\vcW}^{\vcU}}
\newcommand{\vcMU}{{\vcM}^{\vcU}}
\newcommand{\vT}{{{}\vcT}}

\newcommand{\vWU}{\vcW^{\lU}}

\newcommand{\valv}{{\mathtt{val}}}

\newcommand{\vcTV}{\ensuremath{{\vcT_\valv}}}
\newcommand{\vcWV}{\ensuremath{{\vcW_\valv}}}
\newcommand{\vcMV}{\ensuremath{{\vcM_\valv}}}
\newcommand{\vcTUV}{\ensuremath{\vcTU_\valv}}
\newcommand{\vcWUV}{\ensuremath{\vcWU_\valv}}
\newcommand{\vcMUV}{\ensuremath{\vcMU_\valv}}

\newcommand{\vTtp}[1]{{{}\vT_{#1}}}

\newcommand{\vTt}[1]{{}\vT^{#1}}
\newcommand{\vWt}{{{\vcW}}}
\newcommand{\vTUt}[1]{{{}\vcTU_{#1}}}
\newcommand{\vWUt}{{{}\vcWU}}

\newcommand{\crit}[1]{\mathtt{Crit}_{\loc}}

\newcommand{\wmax}[1]{\mathtt{w}^{\sf max}_{#1}}

\newcounter{mylabelcounter}

\makeatletter
\newcommand{\labelAxiom}[2]{%
\hfill{\normalfont\textsc{(#1)}}\refstepcounter{mylabelcounter}
\immediate\write\@auxout{%
  \string\newlabel{#2}{{\unexpanded{\normalfont\textsc{#1}}}{\thepage}{{\unexpanded{\normalfont\textsc{#1}}}}{mylabelcounter.\number\value{mylabelcounter}}{}}
}%
}
\makeatother

\newcommand{\squishlist}[1][$\bullet$]{%
 \begin{list}{#1}
  { \setlength{\itemsep}{0pt}
     \setlength{\parsep}{0pt}
     \setlength{\topsep}{1pt}
     \setlength{\partopsep}{0pt}
     \setlength{\leftmargin}{1.2em}
     \setlength{\labelwidth}{0.5em}
     \setlength{\labelsep}{0.4em} } }
\newcommand{\squishend}{
  \end{list}  }

\newcommand{\changed}[1]{{\color{red!70!black}{#1}}}

\newcommand{\prog}{{P}}

\newcommand{\typ}{\mathit{typ}}
\newcommand{\loc}{{x}}
\newcommand{\loca}{{y}}

\newcommand{\tid}{{\tau}}
\newcommand{\tida}{{\pi}}

\newcommand{\lab}{{l}}

\newcommand{\val}{v}
\newcommand{\vala}{u}
\newcommand{\valset}{V}
\renewcommand{\mod}{o}

\newcommand*\circled[1]{\tikz[baseline=(char.base)]{
            \node[shape=rectangle,rounded corners,draw=red,inner sep=2pt] (char) {\ensuremath{#1}};}}

\newcommand{\toolname}{{\textsc{RSan}}\xspace}
\newcommand{\toolnamefull}{Robustness Sanitizer\xspace}

\makeatletter
\newcommand{\mylabel}[2]{#2\def\@currentlabel{#2}\label{#1}}
\makeatother

\newcommand{\curv}{{\mathtt{c}}}
\newcommand{\acqv}{{\mathtt{a}}}
\newcommand{\relv}{{\mathtt{r}}}

\makeatletter
\newcommand*{\eqassnh}{\mathrel{\rlap{%
                     \raisebox{0.3ex}{$\m@th\cdot$}}%
                     \raisebox{-0.3ex}{$\m@th\cdot$}}%
                     }
\newcommand*{\eqassn}{\,{\eqassnh}{=}\,}
\newcommand*{\sqassn}{\,{\eqassnh}{\sqcup}{=}\,}
\makeatother

\algnewcommand\algorithmicswitch{\textbf{switch}}
\algnewcommand\algorithmiccase{\textbf{case}}
\algnewcommand\algorithmicassert{\texttt{assert}}
\algnewcommand\Assert[1]{\State \algorithmicassert(#1)}%
\algdef{SE}[SWITCH]{Switch}{EndSwitch}[1]{\algorithmicswitch\ #1\ \algorithmicdo}{\algorithmicend\ \algorithmicswitch}%
\algdef{SE}[CASE]{Case}{EndCase}[1]{\algorithmiccase\ #1}{\algorithmicend\ \algorithmiccase}%
\algtext*{EndSwitch}%
\algtext*{EndCase}%

\newcommand{\floc}{\mathtt{f}}

\newcommand{\ctid}[1]{{{}\mathtt{T}_#1}}
\newcommand{\cloc}[1]{\mathtt{
\ifthenelse{\equal{#1}{1}}{x}{
\ifthenelse{\equal{#1}{2}}{y}{
\ifthenelse{\equal{#1}{3}}{z}{
\ifthenelse{\equal{#1}{4}}{w}{
\problem}}}}}}

\newcommand{\creg}[1]{\mathtt{
\ifthenelse{\equal{#1}{1}}{a}{
\ifthenelse{\equal{#1}{2}}{b}{
\ifthenelse{\equal{#1}{3}}{c}{
\ifthenelse{\equal{#1}{4}}{d}{
\ifthenelse{\equal{#1}{5}}{e}{
\ifthenelse{\equal{#1}{6}}{f}{
\problem}}}}}}}}

\xspaceaddexceptions{1}

\newcommand{\callInst}[3]{
\ifthenelse{\equal{#1}{}}
{\ifthenelse{\equal{#3}{}}
{\kw{call}({#2})}
{\kw{call}({#2},{#3})}}
{\ifthenelse{\equal{#3}{}}
{#1 \;{:=}\; \kw{call}({#2})}
{#1 \;{:=}\; \kw{call}({#2},{#3})}}
}

\newcommand{\Time}{\mathsf{Time}}
\newcommand{\ts}{t}
\newcommand{\lTS}{{\color{colorTS} \mathtt{ts}}}

\newcommand{\lTSU}{{\lTS^\lU}}
\newcommand{\epoch}[2]{#1{\texttt{@}}#2}

\newcommand{\tsan}{\textsc{TSan}\xspace}
\newcommand{\ctsan}{\textsc{TSan11}\xspace}
\newcommand{\ctester}{\textsc{C11Tester}\xspace}
\newcommand{\rocker}{\textsc{Rocker}\xspace}

\newcommand{\hsz}{0.75em}

\newcommand{\hsep}{\hspace*{\hsz}}

\newcommand{\bmark}[1]{\texttt{#1}}

\newenvironment{adjblockarray}
 {\begin{lrbox}{\adjblockarraybox}$\begin{blockarray}}
 {\end{blockarray}$\end{lrbox}%
  \raisebox{-1.5ex}[\dimexpr\height-2ex][\dimexpr\depth-2ex]{\usebox{\adjblockarraybox}}}
\newsavebox{\adjblockarraybox}

\newcommand{\exl}[2]{%
\begin{adjblockarray}{c}
  {\textcolor{black}{\scriptstyle #1}}\\
  {\textcolor{black}{\scriptstyle #2}}\\
\end{adjblockarray}
}

\newcommand{\exmx}[2]{%
\begin{adjblockarray}{(c)}
  {\scriptstyle #1}\\
  {\scriptstyle #2}\\
\end{adjblockarray}
}

\usepackage{tagpdf-base}

\newsavebox{\initb}
\sbox{\initb}
{\centering$\inarr{\begin{tikzpicture}
      \node (msg) at (0,.4) {\color{gray}{\emph{initially:}}};
       \node (0x)  at (-.5,0) {$\wlab{}{\cloc1}{0}$};
       \node (0y)  at (.5,0)  {$\wlab{}{\cloc2}{0}$};
      \phantom{\node (wx1) at (-.9,-.75) {$\wlab{}{x}{1}$};}
      \phantom{\node (ry) at (-.9,-1.5) {$\rlab{}{y}{0}$};}
      \phantom{\node (wy1) at (.9,-.75) {$\wlab{}{y}{1}$};}
    \end{tikzpicture}}$
}
\newsavebox{\wx}
\sbox{\wx}
{\centering$\inarr{\begin{tikzpicture}[]
  \node (msg) at (0,.4) {\color{gray}{\emph{after $\assignInst{\cloc1}{1}$:}}};
  \node (0x)  at (-.5,-0) {$\wlab{}{\cloc1}{0}$};
  \node (0y)  at (.5,0)  {$\wlab{}{\cloc2}{0}$};
  \node (wx1) at (-.9,-.75) {$\wlab{}{x}{1}$};
  \phantom{\node (ry) at (-.9,-1.5) {$\rlab{}{y}{0}$};}
  \phantom{\node (wy1) at (.9,-.75) {$\wlab{}{y}{1}$};}

  \draw[po] (0x) edge (wx1);
  \draw[po] (0y) edge (wx1);
  \draw[mo,bend right=10] (0x) edge (wx1);

  \phantom{\draw[po] (0x) edge (wy1);\draw[po] (wx1) edge (ry);\draw[po] (0x) edge (wx1);}
\end{tikzpicture}}$
}
\newsavebox{\wxry}
\sbox{\wxry}
{\centering$\inarr{\begin{tikzpicture}[]
  \node (msg) at (0,.4) {\color{gray}{\emph{after $\assignInst{\creg1}{\cloc2}$:}}};
  \node (0x)  at (-.5,-0) {$\wlab{}{\cloc1}{0}$};
  \node (0y)  at (.5,0)  {$\wlab{}{\cloc2}{0}$};
  \node (wx1) at (-.9,-.75) {$\wlab{}{x}{1}$};
  \node (ry) at (-.9,-1.5) {$\rlab{}{y}{0}$};
  \phantom{\node (wy1) at (.9,-.75) {$\wlab{}{y}{1}$};}

  \draw[po] (0x) edge (wx1); %
  \draw[po] (0y) edge (wx1); %
  \draw[po] (wx1) edge (ry);
  \draw[mo,bend right=10] (0x) edge (wx1);
  \draw[rf,bend right=10] (0y) edge (ry);
\end{tikzpicture}}$
}
\newsavebox{\wxrywy}
\sbox{\wxrywy}
{\centering$\inarr{\begin{tikzpicture}[]
  \node (msg) at (0,.4) {\color{gray}{\emph{after $\assignInst{\cloc2}{1}$:}}};
  \node (0x)  at (-.5,-0) {$\wlab{}{\cloc1}{0}$};
  \node (0y)  at (.5,0)  {$\wlab{}{\cloc2}{0}$};
  \node (wx1) at (-.9,-.75) {$\wlab{}{x}{1}$};
  \node (ry) at (-.9,-1.5) {$\rlab{}{y}{0}$};
  \node (wy1) at (.9,-.75) {$\wlab{}{y}{1}$};

  \draw[po] (0x) edge (wx1) edge (wy1);
  \draw[po] (0y) edge (wx1) edge (wy1);
  \draw[po] (wx1) edge (ry);

  \draw[mo,bend right=10] (0x) edge (wx1);
  \draw[mo,bend left=10] (0y) edge (wy1);
  \draw[rf,bend right=10] (0y) edge (ry);
\end{tikzpicture}}$
}
\newsavebox{\sbex}
\sbox{\sbex}
{
  $\inarrII{\assignInst{\cloc1}{1} \\ \assignInst{\creg1}{\cloc2}}
  {\assignInst{\cloc2}{1} \\ \assignInst{\creg2}{\cloc1}}$
}

\pgfplotsset{%
    discard if/.style 2 args={%
        x filter/.code={%
            \edef\tempa{\thisrow{#1}}
            \edef\tempb{#2}
            \ifx\tempa\tempb
                
            \fi
        }
    },
    discard if not/.style 2 args={%
        x filter/.code={%
            \edef\tempa{\thisrow{#1}}
            \edef\tempb{#2}
            \ifx\tempa\tempb
            \else
                
            \fi
        }
    }
}

\begin{document}

\newcommand{\mytitle}{Dynamic Robustness Verification Against Weak Memory}
\title[\mytitle]{\mytitle}
\iflong\title[\mytitle \ (Extended Version)]{\mytitle\\ (Extended Version)}\fi

\author{Roy Margalit}
\orcid{0000-0001-7266-8681}             %
\affiliation{
  \institution{Tel Aviv University}           %
  \country{Israel}                   %
}
\email{roy.margalit@cs.tau.ac.il}         %

\author{Michalis Kokologiannakis}
\orcid{0000-0002-7905-9739}
\affiliation{
  \institution{ETH Zurich}
  \city{Department of Computer Science, Zurich}
  \country{Switzerland}
}
\email{michalis.kokologiannakis@inf.ethz.ch}

\author{Shachar Itzhaky}
\orcid{0000-0002-7276-7644}
\affiliation{
  \institution{Technion}
  \country{Israel}
}
\email{shachari@cs.technion.ac.il}

\author{Ori Lahav}
\orcid{0000-0003-4305-6998}             %
\affiliation{
  \institution{Tel Aviv University}           %
  \country{Israel}                   %
}
\email{orilahav@tau.ac.il}         %

\begin{abstract}
Dynamic race detection is a highly effective runtime verification technique
for identifying data races by instrumenting and monitoring concurrent program runs.
However, standard dynamic race detection
is incompatible with practical weak memory models;
the added instrumentation introduces extra synchronization,
which masks weakly consistent behaviors and inherently misses certain data races.
In response, we propose
to dynamically verify \emph{program robustness}---a property ensuring that a program exhibits only strongly consistent behaviors.
Building on an existing static decision procedures,
we develop an algorithm for dynamic robustness verification under a C11-style memory model.
The algorithm is based on ``location clocks'', a variant of vector clocks used in standard race detection.
It allows effective and easy-to-apply defense against weak memory on a per-program basis,
which can be combined with race detection that assumes strong consistency.
We implement our algorithm in a tool, called \toolname, and evaluate it across various settings.
To our knowledge, this work is the first to propose and develop dynamic verification of robustness
against weak memory models.
\end{abstract}


\iflong 
\else
\fi

\maketitle

\section{Introduction}%
\label{sec:intro}

Developers are well-aware that %
concurrent code is extremely hard to get right.
With static verification methods not scaling to
real-world programs, dynamic techniques have gained
widespread success and adoption,
and tools like ThreadSanitizer (\tsan) have become essential
parts of mainstream compilers~\cite{Serebryany09,www:tsan}.
Such tools detect \emph{data races}---situations
where two threads concurrently access the same variable,
and at least one of them is writing---commonly leading to
bugs in concurrent programs.
To do so, they instrument the program to track synchronization 
at runtime and use the instrumentation for detecting races in the current run
and for predicting races that will occur in other runs~\cite{Mathur22}.

Seminal work on dynamic race detection assumes that inter-thread
synchronization is achieved using locks or strong synchronization variables (\eg
Java's volatile accesses)~\cite{Bond10,Flanagan09,Savage97}.
More specifically, race detectors allow concurrent accesses to locks and synchronization variables,
but consider races on all other variables as bugs, and report them to the user.
Crucially, they assume that the synchronization variables
follow \emph{sequential consistency} (\SC), \ie accesses to
these variables behave as if they are interleaved, and every read obtains
its value from the latest write to the same variable.

In practice, however, synchronization variables provide guarantees weaker than \SC.
In particular, since 2011, C/C++ has introduced performance-oriented specialized
synchronization accesses, known as \emph{atomic accesses},
with several different strengths that allow behaviors weaker than \SC,
as well as \emph{fences} for fine-grained control
on the synchronization between these accesses.
The semantics of atomic accesses and fences is defined by the C11 \emph{weak (or relaxed)
memory model}~\cite{cppstandard,cstandard,Batty-al:POPL11,scfix}.

In the context of a weak memory model like C11, dynamic race detection is especially valuable.
Indeed, understanding the model and correctly synchronizing C11
code has been proven a difficult task, and 
developers are often advised to forgo the potential performance
benefits of C11, and instead rely on the
``good old'' SC accesses and locks~\cite{Boehm:2008}.
In turn, dynamic race detection would automatically identify subtle issues
that lead to data races, and guide programmers in strengthening
synchronization variables or adding fences to ensure sufficient synchronization.

Unfortunately, developing dynamic race detectors for languages 
with weak memory models is highly challenging 
due to what we term the  ``observer effect'': 
the instrumentation code that monitors for data races requires its own
synchronization mechanisms to avoid internal data races,
which inadvertently masks non-SC behaviors of the original program.
We demonstrate the consequences of the observer effect with the
following example.

\begin{example}
For the program on the right, the behavior where both reads read $0$ is forbidden
under
\begin{wrapfigure}[4]{r}{.15\textwidth}
$\inarrII{\assignInst{\cloc1}{1} \\ \assignInst{\creg1}{\cloc2}}
{\assignInst{\cloc2}{1} \\ \assignInst{\creg2}{\cloc1}}$
\end{wrapfigure}
\SC, but allowed by C11---and observable on standard %
laptops---if we use weakly consistent atomic accesses.
However, as soon as we use a dynamic race detector like \tsan on this
program, the annotated outcome no longer manifests. 
While it is known that \tsan does not fully handle the C11 weak memory model
(including soundness issues\footnote{See, \eg \url{https://gcc.gnu.org/bugzilla/show_bug.cgi?id=97868} and \url{https://github.com/google/sanitizers/issues/1415} (accessed November 2024)}),
it is perhaps less recognized that running \tsan on a program 
can mask certain behaviors that the program might otherwise exhibit.
A real-world case where \tsan
misses a race due to this problem, is the well-known (broken) implementation of Dekker's mutual
exclusion protocol~\cite{EWD123} using C11 release/acquire atomics.
When run with \tsan, the race within the intended critical section is never revealed.
\label{ex:obs}
\end{example}

\vspace*{-4pt}
Weak behaviors like the one above remain hidden during
monitoring, but they will resurface in production code, as soon as the
instrumentation is removed.
In fact, even logging the executed accesses for a ``post-mortem''
analysis introduces extra synchronization: accessing the shared log
requires a memory barrier, which again restricts the behaviors a
program would otherwise exhibit.

To address this issue, one must essentially execute the program while
simulating the C11 constraints.
For instance, a read can no longer simply retrieve the most recent
written value, but rather the instrumentation should (randomly)
select a value of some write permitted to read-from by the memory model.
Such an approach leads to
a clear trade-off between completeness---\ie
potentially detecting \emph{all} racy behaviors---and
performance/memory overhead compared to standard execution.
Existing methods either significantly sacrifice
completeness~\cite{Lidbury17} or incur high memory overhead
proportional to the execution length~\cite{Luo21}.

In this paper, we propose a novel approach to the problem of dynamic
verification of concurrent programs under a weak memory model.
Our key idea is to accompany dynamic race detection that assumes SC with a
dynamic verification of \emph{robustness}.
Robustness is a general correctness criterion for concurrent
programs under weak memory models that requires that all behaviors allowed
under the weak model are also allowed under SC\@.
It guarantees that verification in general, and dynamic race detection in particular,
may ignore weak memory behaviors and work under the illusion of SC\@.
Robustness does not rule out the use of weakly consistent atomics.
Indeed, in a variety of cases, one can use weak accesses without making weak behaviors
externally observable, thereby obtaining the best of both worlds: SC semantics
and efficient implementations.

While a large body of previous work (see \cref{sec:related})
has developed verification methods and tools for
robustness,
these methods are static, and do not scale to real-world full-blown programs.
The static methods, however, reveal an important benefit of
robustness, compared to other correctness notions: robustness
is reducible to the absence of a certain shape \emph{in executions under SC}.
This property makes robustness an ideal candidate for dynamic
verification, as robustness (and non-robustness) can be checked while ignoring weak memory
behaviors, and is thus immune to the ``observer effect''.

Based on the above observation, we develop the first dynamic
verification method for program robustness against a C11-style memory model
(we target \RLX, a variant of C11 from~\cite{popl21}, see \cref{sec:rc20}).
This means that we run the program,
and check that no atomic access in the current run may exhibit a weak behavior according to the memory model,
\ie a load reading a stale value or a write becoming stale although it is executed last.
Then, we accompany our technique with a dynamic race detector that
assumes \SC for detecting data races on \emph{non-atomic} accesses.
Dealing with a language-level model, we note that non-robustness does not imply the occurrence of weak behaviors, 
but rather signifies that such behaviors are permitted by the formal model
and can potentially manifest when using specific compilers on particular architectures.

To dynamically verify robustness, we enhance
the instrumentation used by race detectors to 
detect \emph{non-robustness witnesses}.
The latter, defined in \citet{pldi19,popl21}, form our analog of data races.
However, the instrumentation used in 
these works for \emph{static} detection of non-robustness witnesses has significant drawbacks when used \emph{dynamically}
(see \cref{sec:BM_bad}).
To solve this, we develop an instrumentation based on 
\emph{location clocks}, an adaptation of \emph{vector clocks}~\cite{fidge1988timestamps,VirtTimeGlobStates}
that assigns timestamps to memory locations, rather than threads.
The instrumentation based on location clocks is better suited for dynamic analysis.
Specifically, it reduces contention on locks for maintaining the instrumentation,
and thus retains more of the concurrency in the input program.
Moreover, it enables a \emph{predictive} analysis that can identify non-robustness witnesses
by observing executions that do not directly serve as such witnesses.
The proposed technique is sound (\ie never reports spurious robustness violations)
and complete (\ie if a program is not robust, there exists at least one schedule under
which the analysis reports a violation), 
and its memory consumption is practically independent of the length of the run.

We implement our approach on top of \tsan in a tool called \toolname
(\toolnamefull), and evaluate it on multiple examples and case studies.
\toolname scales to real-world codebases with thousands of lines of code.
Upon detecting a robustness violation,
\toolname reports the offending instructions to the user,
who can then prevent the corresponding weak behavior
by strengthening some accesses and/or inserting fences.
With its simplicity of use and applicability to C/C++ as-is,
we envision \toolname being used in a defensive programming
style, enabling developers to improve performance of their code
without being experts in weak memory models.

Finally, we also develop a race detection algorithm for C/C++11 
(specifically, its \RLX variant), 
and implement it in \tsan's infrastructure. 
Although robustness (for atomics) permits the race detector to assume \SC semantics,
the detection algorithm must still account for the model's definition of data races on non-atomics, 
which are determined by the way atomics and fences are used.

The rest of the paper is organized as follows:

\begin{itemize}
\item[\cref{sec:pre}] We recap (a fragment of) C11
  and its robustness verification technique from prior work.
\item[\cref{sec:algorithm}] We introduce our dynamic-analysis
  algorithm based on location clocks.
\item[\cref{sec:extensions}] We present extensions to the basic algorithm: covering the full \RLX model
and a technique to avoid reporting common benign robustness violations.
\item[\cref{sec:impl},\cref{sec:eval}] We discuss the implementation of our approach and evaluate our implementation.
\item[\cref{sec:related}] We discuss related work and outline potential future directions.
\end{itemize}
\iflong
Appendices \ref{app:rlx} and \ref{app:race} provide the complete location-clock algorithm
and the modified race detection algorithm.
\else
Appendices with the complete location-clock algorithm
and the modified race detection algorithm are included in \cite{appendix}.
\fi
The \toolname tool and the examples used for its evaluation are available in the accompanying artifact
available at \url{https://doi.org/10.5281/zenodo.15002567}.

\newcommand{\figrob}{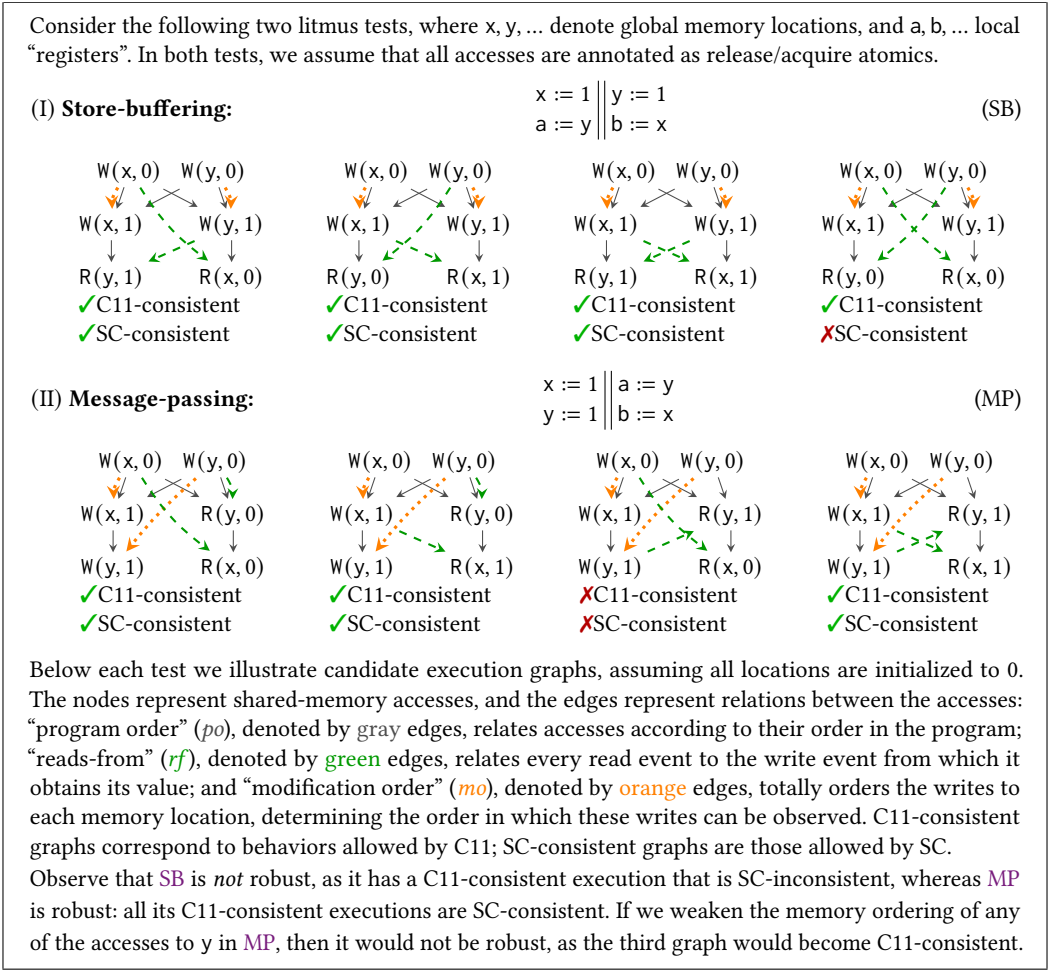
\begin{figure*}
\begin{mdframed}
{\small
Consider the following two litmus tests, where $\cloc1,\cloc2,\ldots$ denote global memory locations,
and $\creg1, \creg2,\ldots$ local ``registers''. %
In both tests, we assume that all accesses are annotated as release/acquire atomics.

\noindent
\begin{flalign}
        &(\mathrm{I}) \textbf{ Store-buffering:} &
\inarrII{\assignInst{\cloc1}{1} \\ \assignInst{\creg1}{\cloc2}}
{\assignInst{\cloc2}{1} \\ \assignInst{\creg2}{\cloc1}} &&
        \tag{SB}\label{prog:sb}
    \end{flalign}

\small
\noindent
\hfill
\begin{tikzpicture}[yscale=0.7,xscale=0.8]
  \node (0x)  at (-0.7,2) {$\evlab{\lW}{}{\cloc1}{0}$};
  \node (0y)  at (0.7,2) {$\evlab{\lW}{}{\cloc2}{0}$};
  \node (11)  at (-1,1) {$\evlab{\lW}{}{\cloc1}{1}$ };
  \node (12)  at (-1,0) {$\evlab{\lR}{}{\cloc2}{1}$ };
  \node (21)  at (1,1) {$\evlab{\lW}{}{\cloc2}{1}$ };
  \node (22)  at (1,0) {$\evlab{\lR}{}{\cloc1}{0}$ };
  \draw[po] (11) edge (12);
  \draw[po] (21) edge (22);
  \draw[po] (0x) edge (11) edge (21);
  \draw[po] (0y) edge (11) edge (21);
    \draw[mo,bend right=15] (0x) edge (11);
      \draw[rf] (21) edge (12);
  \draw[mo,bend left=15] (0y) edge (21);
      \draw[rf,bend right=10] (0x) edge (22);
        \node  at (-0.2,-0.8) {$\inarr{\text{\cmark C11-consistent}\\ \text{\cmark SC-consistent}}$};
\end{tikzpicture}
\hfill
\begin{tikzpicture}[yscale=0.7,xscale=0.8]
  \node (0x)  at (-0.7,2) {$\evlab{\lW}{}{\cloc1}{0}$};
  \node (0y)  at (0.7,2) {$\evlab{\lW}{}{\cloc2}{0}$};
  \node (11)  at (-1,1) {$\evlab{\lW}{}{\cloc1}{1}$ };
  \node (12)  at (-1,0) {$\evlab{\lR}{}{\cloc2}{0}$ };
  \node (21)  at (1,1) {$\evlab{\lW}{}{\cloc2}{1}$ };
  \node (22)  at (1,0) {$\evlab{\lR}{}{\cloc1}{1}$ };
  \draw[po] (11) edge (12);
  \draw[po] (21) edge (22);
  \draw[po] (0x) edge (11) edge (21);
  \draw[po] (0y) edge (11) edge (21);
    \draw[mo,bend right=15] (0x) edge (11);
      \draw[rf,bend left=10] (0y) edge (12);
  \draw[mo,bend left=15] (0y) edge (21);
      \draw[rf] (11) edge (22);
        \node  at (-0.2,-0.8) {$\inarr{\text{\cmark C11-consistent}\\ \text{\cmark SC-consistent}}$};
\end{tikzpicture}
\hfill
\begin{tikzpicture}[yscale=0.7,xscale=0.8]
  \node (0x)  at (-0.7,2) {$\evlab{\lW}{}{\cloc1}{0}$};
  \node (0y)  at (0.7,2) {$\evlab{\lW}{}{\cloc2}{0}$};
  \node (11)  at (-1,1) {$\evlab{\lW}{}{\cloc1}{1}$ };
  \node (12)  at (-1,0) {$\evlab{\lR}{}{\cloc2}{1}$ };
  \node (21)  at (1,1) {$\evlab{\lW}{}{\cloc2}{1}$ };
  \node (22)  at (1,0) {$\evlab{\lR}{}{\cloc1}{1}$ };
  \draw[po] (11) edge (12);
  \draw[po] (21) edge (22);
  \draw[po] (0x) edge (11) edge (21);
  \draw[po] (0y) edge (11) edge (21);
    \draw[mo,bend right=15] (0x) edge (11);
      \draw[rf] (21) edge (12);
  \draw[mo,bend left=15] (0y) edge (21);
      \draw[rf] (11) edge (22);
        \node  at (-0.2,-0.8) {$\inarr{\text{\cmark C11-consistent}\\ \text{\cmark SC-consistent}}$};
\end{tikzpicture}
\hfill
\begin{tikzpicture}[yscale=0.7,xscale=0.8]
  \node (0x)  at (-0.7,2) {$\evlab{\lW}{}{\cloc1}{0}$};
  \node (0y)  at (0.7,2) {$\evlab{\lW}{}{\cloc2}{0}$};
  \node (11)  at (-1,1) {$\evlab{\lW}{}{\cloc1}{1}$ };
  \node (12)  at (-1,0) {$\evlab{\lR}{}{\cloc2}{0}$ };
  \node (21)  at (1,1) {$\evlab{\lW}{}{\cloc2}{1}$ };
  \node (22)  at (1,0) {$\evlab{\lR}{}{\cloc1}{0}$ };
  \draw[po] (11) edge (12);
  \draw[po] (21) edge (22);
  \draw[po] (0x) edge (11) edge (21);
  \draw[po] (0y) edge (11) edge (21);
    \draw[mo,bend right=15] (0x) edge (11);
      \draw[rf,bend left=10] (0y) edge (12);
  \draw[mo,bend left=15] (0y) edge (21);
      \draw[rf,bend right=10] (0x) edge (22);
        \node  at (-0.2,-0.8) {$\inarr{\text{\cmark C11-consistent}\\ \text{\xmark SC-consistent}}$};
\end{tikzpicture}
\hfill

\noindent
\begin{flalign}
        &(\mathrm{II}) \textbf{ Message-passing:} &
\inarrII{\assignInst{\cloc1}{1} \\ \assignInst{\cloc2}{1}}
{\assignInst{\creg1}{\cloc2} \\ \assignInst{\creg2}{\cloc1}} &&
        \tag{MP}\label{prog:mp}
    \end{flalign}

\small
\noindent
\hfill
\begin{tikzpicture}[yscale=0.7,xscale=0.8]
  \node (0x)  at (-0.7,2) {$\evlab{\lW}{}{\cloc1}{0}$};
  \node (0y)  at (0.7,2) {$\evlab{\lW}{}{\cloc2}{0}$};
  \node (11)  at (-1,1) {$\evlab{\lW}{}{\cloc1}{1}$ };
  \node (12)  at (-1,0) {$\evlab{\lW}{}{\cloc2}{1}$ };
  \node (21)  at (1,1) {$\evlab{\lR}{}{\cloc2}{0}$ };
  \node (22)  at (1,0) {$\evlab{\lR}{}{\cloc1}{0}$ };
  \draw[po] (11) edge (12);
  \draw[po] (21) edge (22);
  \draw[po] (0x) edge (11) edge (21);
  \draw[po] (0y) edge (11) edge (21);
    \draw[mo,bend right=15] (0x) edge (11);
      \draw[rf,bend left=15] (0y) edge (21);
  \draw[mo,bend right=5] (0y) edge (12);
      \draw[rf,bend right=10] (0x) edge (22);
              \node  at (-0.2,-0.8) {$\inarr{\text{\cmark C11-consistent}\\ \text{\cmark SC-consistent}}$};
\end{tikzpicture}
\hfill
\begin{tikzpicture}[yscale=0.7,xscale=0.8]
  \node (0x)  at (-0.7,2) {$\evlab{\lW}{}{\cloc1}{0}$};
  \node (0y)  at (0.7,2) {$\evlab{\lW}{}{\cloc2}{0}$};
  \node (11)  at (-1,1) {$\evlab{\lW}{}{\cloc1}{1}$ };
  \node (12)  at (-1,0) {$\evlab{\lW}{}{\cloc2}{1}$ };
  \node (21)  at (1,1) {$\evlab{\lR}{}{\cloc2}{0}$ };
  \node (22)  at (1,0) {$\evlab{\lR}{}{\cloc1}{1}$ };
  \draw[po] (11) edge (12);
  \draw[po] (21) edge (22);
  \draw[po] (0x) edge (11) edge (21);
  \draw[po] (0y) edge (11) edge (21);
    \draw[mo,bend right=15] (0x) edge (11);
      \draw[rf,bend left=15] (0y) edge (21);
  \draw[mo,bend right=5] (0y) edge (12);
      \draw[rf] (11) edge (22);
              \node  at (-0.2,-0.8) {$\inarr{\text{\cmark C11-consistent}\\ \text{\cmark SC-consistent}}$};
\end{tikzpicture}
\hfill
\begin{tikzpicture}[yscale=0.7,xscale=0.8]
  \node (0x)  at (-0.7,2) {$\evlab{\lW}{}{\cloc1}{0}$};
  \node (0y)  at (0.7,2) {$\evlab{\lW}{}{\cloc2}{0}$};
  \node (11)  at (-1,1) {$\evlab{\lW}{}{\cloc1}{1}$ };
  \node (12)  at (-1,0) {$\evlab{\lW}{}{\cloc2}{1}$ };
  \node (21)  at (1,1) {$\evlab{\lR}{}{\cloc2}{1}$ };
  \node (22)  at (1,0) {$\evlab{\lR}{}{\cloc1}{0}$ };
  \draw[po] (11) edge (12);
  \draw[po] (21) edge (22);
  \draw[po] (0x) edge (11) edge (21);
  \draw[po] (0y) edge (11) edge (21);
    \draw[mo,bend right=15] (0x) edge (11);
      \draw[rf] (12) edge (21);
  \draw[mo,bend right=5] (0y) edge (12);
      \draw[rf,bend right=10] (0x) edge (22);
              \node  at (-0.2,-0.8) {$\inarr{\text{\xmark C11-consistent}\\ \text{\xmark SC-consistent}}$};
\end{tikzpicture}
\hfill
\begin{tikzpicture}[yscale=0.7,xscale=0.8]
  \node (0x)  at (-0.7,2) {$\evlab{\lW}{}{\cloc1}{0}$};
  \node (0y)  at (0.7,2) {$\evlab{\lW}{}{\cloc2}{0}$};
  \node (11)  at (-1,1) {$\evlab{\lW}{}{\cloc1}{1}$ };
  \node (12)  at (-1,0) {$\evlab{\lW}{}{\cloc2}{1}$ };
  \node (21)  at (1,1) {$\evlab{\lR}{}{\cloc2}{1}$ };
  \node (22)  at (1,0) {$\evlab{\lR}{}{\cloc1}{1}$ };
  \draw[po] (11) edge (12);
  \draw[po] (21) edge (22);
  \draw[po] (0x) edge (11) edge (21);
  \draw[po] (0y) edge (11) edge (21);
    \draw[mo,bend right=15] (0x) edge (11);
      \draw[rf] (12) edge (21);
  \draw[mo,bend right=5] (0y) edge (12);
      \draw[rf] (11) edge (22);
                    \node  at (-0.2,-0.8) {$\inarr{\text{\cmark C11-consistent}\\ \text{\cmark SC-consistent}}$};
\end{tikzpicture}

Below each test we illustrate candidate execution graphs,
assuming all locations are initialized to $0$.
The nodes represent shared-memory accesses, and
the edges represent relations between the accesses:
``program order'' ($\po$), denoted by {\color{colorPO}gray} edges,
relates accesses according to their order in the program;
``reads-from'' ($\rf$), denoted by {\color{colorRF}green} edges,
relates every read event to the write event from which it obtains its value; and
``modification order'' ($\mo$), denoted by {\color{colorMO}orange} edges,
totally orders the writes to each memory location, determining the order in which
these writes can be observed.
C11-consistent graphs correspond to behaviors allowed by C11;
\SC-consistent graphs are those allowed by SC\@.

Observe that~\ref{prog:sb} is \emph{not} robust, as it has a C11-consistent execution that is \SC-inconsistent,
whereas~\ref{prog:mp} is robust: all its C11-consistent executions are \SC-consistent.
If we weaken the memory ordering of any of the accesses to $\cloc2$ in~\ref{prog:mp},
then it would not be robust, as the third graph would become C11-consistent.}
\end{mdframed}
\caption{Examples of consistency and robustness.}%
\label{fig:intro_rob}
\end{figure*}}

\section{Preliminaries}%
\label{sec:pre}

We begin by providing the necessary background on the C11 model
(\cref{sec:c11}), the robustness criterion (\cref{sec:robustness}),
and the static verification approach previously proposed for
robustness (\cref{sec:robustness_verification}).
\Cref{fig:intro_rob} presents simple illustrative examples,
which we reuse throughout this section.

Atomic accesses in C11 have memory orderings (``relaxed'',
``release'', ``acquire'', and ``sequentially consistent''), with the
weak orderings being better-performant on multicore hardware compared
to sequentially consistent ones.
To simplify the presentation, in this section we only discuss \emph{release/acquire write and read accesses},
a well-studied fragment of C11~\cite{sra}.
Our full development and the accompanying tool support a
broader subset of C11, which we discuss in \cref{sec:extensions}.

\figrob

\subsection{The Release/Acquire Memory Model}%
\label{sec:c11}

The C11 model is declarative,
capturing program behavior in \emph{execution graphs}---structures
that track several partial orders on memory accesses,
using them to limit the allowed return values of reads.
Nodes in these graphs are called \emph{events}, and correspond to
memory accesses. Edges denote various relations between the accesses.
Let us now formalize these notions.

\paragraph{Domains}
We let $ \Loc=\set{\cloc{1},\cloc{2},\ldots}$ be a finite set of locations,
$\Val$ be a set of values that contains $0$, which we use as the initial value,
and $\Tid=\set{\ctid{1},\ctid{2},\ldots}$ be a finite set of thread identifiers.
In our code snippets, threads are numbered from left to right, 
with the leftmost thread labeled $\ctid1$.

\paragraph{Events}
An \emph{event} $e$ %
is a tuple $\tup{\tid,s,\lab}$,
where $\tid \in \Tid \uplus \set{\bot}$ is a thread identifier (or $\bot$ for initialization events),
$s \in \N$ is a serial number inside each thread,
and $\lab \in \Lab$ is a label,
which can be either
	$\rlab{}{\loc}{\val}$ (read)
	or $\wlab{}{\loc}{\val}$ (write)
	with $\loc\in \Loc$ and $\val\in \Val$.
The functions $\lTID$, $\lLOC$, and $\lVAL$
return the thread identifier,
location, and value of a given event.
We let $\sR$ and $\sW$ denote the sets of all read and write events, respectively.
We  %
employ sub- and superscripts to restrict sets of events,
\eg $\sW_\loc=\set{w\in \sW \st \lLOC(w)=\loc}$
and $E^\tid= \set{e\in E \st \lTID(e)=\tid}$
for a set $E$ of events.

\paragraph{Execution Graphs}
An \emph{execution} graph $G$
is a tuple $\tup{E,\rf,\mo}$, where:
\begin{itemize}
\item $E$ is a finite set of events such that $\Init\suq E$,
  where $\Init \defeq \set{\tup{\bot,0,\wlab{}{\loc}{0}} \st \loc \in \Loc}$ contains an initial write per location.
We require that $\lTID(e)\neq \bot$ for every $e\in E\setminus \Init$,
and that no two different events in $E$ have the same thread identifier and serial number.
\item $\rf$ is a ``reads-from'' mapping that determines the write event from which each read reads its value.
Formally, the following should hold:
\begin{itemize}
\item If $\tup{w,r}\in \rf$, then $w\in E \cap \sW$, $r\in E \cap \sR$, $\lLOC(w)=\lLOC(r)$, and $\lVAL(w)=\lVAL(r)$.
\item For every $r \in E \cap \sR$, there exists exactly one write event $w$ such that $\tup{w,r}\in \rf$.
\end{itemize}
\item $\mo$ is a ``modification order'' that totally orders the writes to each location.
Formally, $\mo$ is a disjoint union of relations $\set{\mo_\loc}_{\loc\in\Loc}$,
such that each $\mo_\loc$ is a strict total order on $E \cap \sW_\loc$.
\end{itemize}
We denote the components of $G$ by $G.\lE$, $G.\lRF$, and $G.\lMO$.
We use $G.\lPO$ (``program order'') to denote the (partial) order on $G.\lE$ in which
initialization events (with $\lTID(e)=\bot$) precede all other events,
and events of the same thread are totally ordered by their serial numbers.
For a set $E'$ of events, we write $G.E'$ for $G.\lE \cap E'$ (\eg $G.\sW=G.\lE\cap \sW$).

\paragraph{Derived Relations}
We use two derived relations in execution graphs:\footnote{
Given a relation $R$, $R^?$ and $R^+$ denote its reflexive and transitive closures.
The inverse of $R$ is denoted by $R^{-1}$,
and the (left) composition of two relations $R_1,R_2$ is denoted by $R_1\seq R_2$.}

\vspace*{-2.5\multicolsep}
\begin{multicols}{2}
  \begin{equation}
G.\lHB \defeq (G.\lPO \cup G.\lRF)^+ \tag{\emph{happens-before}} 
  \end{equation}\break
  \begin{equation}
 G.\lFR \defeq G.\lRF^{-1}\seq G.\lMO  \tag{\emph{from-read}}
  \end{equation}
\end{multicols}

\vspace*{-0.75\multicolsep}
\noindent
Happens-before represents causality among events.
The ``from-read'' relation holds between a read event $r$ and a write event $w$
when $r$ reads from a write $w'$ that is before $w$ in the modification order.

\paragraph{From Programs to Execution Graphs}
A concurrent program $\prog$ induces a labeled transition system (LTS),
where each transition is labeled by a pair $\tup{\tid,\lab} \in \Tid \times \Lab$.
This LTS is independent of memory consistency and includes a non-deterministic choice among all
possible values in $\Val$ for every read instruction in $\prog$.
Its runs induce \emph{candidate execution graphs} of $\prog$,
where the events of each thread are determined according to the labels along the run,
and the reads-from relations and modification orders are arbitrary.
We note that each graph corresponds to a \emph{prefix} of a run of $\prog$.
The full definitions, which are standard, are omitted (see, \eg~\cite{pldi19}).

\paragraph{Consistency}
Among all candidate execution graphs of a program,
only \emph{consistent} ones correspond to allowed program behaviors.
For the fragment discussed here, an execution graph $G$ is \emph{consistent} if the following hold:
\begin{itemize}
\item $G.\lMO\seq G.\lHB$ is irreflexive.\ \labelAxiom{write coherence}{ax:wcoh}
\item $G.\lFR\seq G.\lHB$ is irreflexive.\ \labelAxiom{read coherence}{ax:rcoh}
\item $G.\lPO \cup G.\lRF$ is acyclic.\ \labelAxiom{acyc-po-rf}{ax:hb}
\end{itemize}

For instance, the third execution graph of~\ref{prog:mp} in \cref{fig:intro_rob} is inconsistent,
as it violates~\ref{ax:rcoh}.
Specifically, there is a $G.\lHB$ relation from the event labeled $\lW(x,1)$
to the event labeled $\lR(x,0)$, while $G.\lFR$ between these events points in the opposite direction.

\subsection{Robustness}%
\label{sec:robustness}

Following previous work, we define robustness using execution graphs.
A program $\prog$ is considered \emph{robust}
against a (declarative) memory model $M$ if every $M$-consistent execution graph
generated by $\prog$ is also \SC-consistent.
\SC-consistency of an execution graph $G$ means that
$G.\lPO \cup G.\lRF \cup G.\lMO$ can be linearized
into 
a sequence of events such that every $G.\lRF$-edge
connects each read $r$ with the last write $w$ to the same location that occurs before $r$.
An equivalent definition from~\cite{herding-cats} is provided next.

\begin{definition}%
\label{def:sc_cons}
An execution graph $G$ is \emph{\SC-consistent} if $G.\lHBSC$ is irreflexive,
where:
\begin{align*}
  G.\lHBSC &\defeq (G.\lPO \cup G.\lRF \cup G.\lMO \cup G.\lFR)^+ \tag{\emph{\SC-happens-before}}
\end{align*}
\end{definition}
For instance, the fourth execution graph of~\ref{prog:sb} in \cref{fig:intro_rob} is \SC-inconsistent.
It has a cycle following $\lFR$, $\lPO$, $\lFR$, and $\lPO$
between the events labelled $\lR(\cloc1,0)$, $\lW(\cloc1,1)$, $\lR(\cloc2,0)$, $\lW(\cloc2,1)$, and back to $\lR(\cloc1,0)$.

The \SC-happens-before relation, $\lHBSC$, holds between events $e_1$ and $e_2$
iff the operation associated with $e_1$ has to be executed before the one associated with $e_2$
in every operational \SC run that generates $G$,
\ie in every operational \SC run in which threads execute the operations
associated with the events in $G$ following their order in $G.\lPO$;
every read $r$ reads its value from the write $w$ for which $\tup{w,r}\in G.\lRF$;
and writes to each location are executed following their order in $G.\lMO$.

\subsection{Robustness Verification}%
\label{sec:robustness_verification}

\citet{pldi19} studied the problem of static verification of robustness for a given program.
The challenge lies in the fact that, although one typically assumes a finite-state programs (as we do here),
such programs may still have loops, and thus have infinitely many consistent execution graphs,
making it impossible to enumerate them all.
Still, \citet{pldi19} provided a decision procedure for this problem
by reducing it to reachability in a \emph{finite}-state machine.
Next, we summarize the key insights from their work.

\paragraph{Non-Robustness Witnesses}
The first key idea is that robustness can be verified by considering only runs under \SC{}\@.
In essence, when searching for executions that deviate from \SC,
it suffices to consider ``borderline executions'': execution graphs of the program that are still \SC-consistent,
but the next step of the program that will add one more event to the graph may
result in an \SC-inconsistent graph while maintaining the consistency conditions of the weak memory model.
Such executions are called \emph{non-robustness witnesses} of the given program.
For the formal definition,
we let $G.\wmax{\loc} \defeq \max_{G.\lMO} G.\sW_\loc$
(the most recent write to location $\loc$ in an execution graph $G$).
Then, a non-robustness witness is a configuration obtained during a run under \SC in which 
some thread $\tid$ is about to access
(either read from or write to) a location $\loc$
and the following hold for the (\SC-consistent) execution graph $G$ that was generated until this point:
\begin{enumerate}
\item The next access of $\tid$ races with $G.\wmax{\loc}$,
\ie thread $\tid$ is not $G.\lHB$-``aware'' of $G.\wmax{\loc}$.
Formally, this means that there is no $G.\lHB$ from $G.\wmax{\loc}$
to some $e \in G.\lE$ with $\lTID(e)=\tid$.
\item The next access of $\tid$ cannot have been executed before $G.\wmax{\loc}$,
\ie it must be executed after $G.\wmax{\loc}$ in any \SC run that would generate $G$.
Formally, this means that we have $G.\lHBSC$ from $G.\wmax{\loc}$
to some $e \in G.\lE$ with $\lTID(e)=\tid$.
\end{enumerate}

The first condition ensures that under the C11 memory model
thread $\tid$ can: 
\begin{enumerate*}[label=(\roman*)]
\item read a stale value from $\loc$ if it tries to read $\loc$; and
\item write a value that is not overwriting the last write to $\loc$ if it writes to $\loc$.
\end{enumerate*}
The second condition ensures that no execution under \SC can generate an execution graph
where such access is performed.
\citet{pldi19} show that a program $\prog$ is not robust
\emph{iff} there is some run of $\prog$ that reaches a non-robustness witness.

\begin{example}%
\label{ex:SB_witness}
Consider again the~\ref{prog:sb} program from \cref{fig:intro_rob}.
Under \SC,  if we execute $\ctid1$ first, and then the write instruction of $\ctid2$,
we reach the following ``borderline'' execution graph:

\noindent
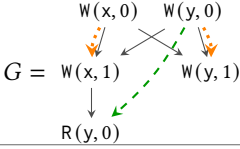
\begin{minipage}{0.3\textwidth}
\[G=\vcenter{\hbox{\begin{tikzpicture}[yscale=0.8,xscale=0.8,font=\footnotesize]
  \node (0x)  at (-0.7,2) {$\evlab{\lW}{}{\cloc1}{0}$};
  \node (0y)  at (0.7,2) {$\evlab{\lW}{}{\cloc2}{0}$};
  \node (11)  at (-1,1) {$\evlab{\lW}{}{\cloc1}{1}$ };
  \node (12)  at (-1,0) {$\evlab{\lR}{}{\cloc2}{0}$ };
  \node (21)  at (1,1) {$\evlab{\lW}{}{\cloc2}{1}$ };
  \draw[po] (11) edge (12);
  \draw[po] (0x) edge (11) edge (21);
  \draw[po] (0y) edge (11) edge (21);
    \draw[mo,bend right=15] (0x) edge (11);
      \draw[rf,bend left=10] (0y) edge (12);
  \draw[mo,bend left=15] (0y) edge (21);
\end{tikzpicture}}}\]
\end{minipage}\hfill
\begin{minipage}{0.68\textwidth}
With this graph we obtain a non-robustness witness for~\ref{prog:sb}.
Indeed,
\begin{enumerate*}[label=(\roman*)]
\item $G.\wmax{\cloc1} =\max_{G.\lMO} G.\sW_\cloc1$ is the write event labeled $\lW(\cloc1,1)$;
\item the next access of $\ctid2$ accesses location $\cloc1$;
\item there is no $G.\lHB$ from $G.\wmax{\cloc1}$ to $\ctid2$;
but \item there is $G.\lHBSC$ from $G.\wmax{\cloc1}$ to $\ctid2$
(via an $G.\lFR$-edge from the event labeled $\lR(\cloc2,0)$ to the one labeled $\lW(\cloc2,1)$).
\end{enumerate*}
\end{minipage}

\smallskip
\end{example}

\paragraph{Bounded Instrumentation}
The second key observation is that
for detecting non-robustness witnesses,
it is unnecessary to retain the entire \SC-consistent execution graph generated thus far.
Instead, maintaining a constant-size summary of the graph is sufficient.
Concretely, at each step, we need to know for every thread $\tid$ and location $\loc$:
\begin{enumerate*}[label=(\roman*)]
\item whether there is $G.\lHB$ from $G.\wmax{\loc}$ to thread $\tid$;
and \item whether there is $G.\lHBSC$ from $G.\wmax{\loc}$ to thread  $\tid$.
\end{enumerate*}
This information can be stored in two ``Boolean matrices'' (BMs, for short)
$\bmT_\HB,\bmT_\SC : \Tid \to 2^\Loc$, such that:
\begin{align*}
\loc\in \bmT_\HB(\tid) &\Leftrightarrow  \exists e \in G.\lE^\tid \cup \Init \ldotp \tup{G.\wmax{\loc},e} \in G.\lHB^? \\
\loc\in \bmT_\SC(\tid) & \Leftrightarrow \exists e \in G.\lE^\tid\cup \Init \ldotp \tup{G.\wmax{\loc},e} \in G.\lHBSC^?
\end{align*}
In turn, three additional BMs,
$\bmW_\HB, \bmW_\SC, \bmM_\SC : \Loc \to 2^\Loc$,
carry sufficient information for faithfully updating $\bmT_\HB$ and $\bmT_\SC$ in every execution step:
\begin{align*}
\loc\in \bmW_\HB(\loca) &\Leftrightarrow \tup{G.\wmax{\loc},G.\wmax{\loca}}\in G.\lHB^?  & &&
\loc\in \bmW_\SC(\loca) &\Leftrightarrow \tup{G.\wmax{\loc},G.\wmax{\loca}}\in G.\lHBSC^? \\
&&&& \loc\in \bmM_\SC(\loca) &\Leftrightarrow  \exists e \in G.\lE_\loca \ldotp \tup{G.\wmax{\loc},e} \in G.\lHBSC^?
\end{align*}
Initially, $\bmT_\HB = \bmT_\SC = \lambda \tid\ldotp \Loc$
and $\bmW_\HB = \bmW_\SC = \bmM_\SC = \lambda \loc\ldotp \set{\loc}$.
The maintenance of the instrumentation with each access to a shared variable is given in \cref{fig:bm:transitions}.
The rows correspond to operations performed by the program,
and each column describes how the instrumentation is updated.
When thread $\tid$ writes to location $\loc$, the location $\loc$ is added to $\bmT_\HB(\tid)$,
and removed from $\bmT_\HB(\tida)$ for all other threads $\tida\neq\tid$.
Indeed, after the write, $\tid$ is aware of the most recent write to $\loc$ but no other thread is aware of it.
In addition, we assign $\bmW_\HB(\loc) \eqassn \bmT_\HB(\tid)$ as the thread releases its view via the write to $\loc$,
and for all other locations $\loca \neq \loc$, the location $\loc$ is removed from $\bmW_\HB(\loca)$.
In turn, when $\tid$ reads from location $\loc$, it acquires the last view released to $\loc$, so we add
$\bmW_\HB(\loc)$ to $\bmT_\HB(\tid)$.
The \SC-BMs update rules are more involved, but are generally similar.
In particular, in order to track the $\lFR$ order (which is a part of $\lHBSC$),
read events propagate the thread view $\bmT_\SC$ to $\bmM_\SC$,
which is acquired by later writes to the same location.

\begin{figure}[t]
\centering
\resizebox{0.75\width}{!}{$
\begin{array}{c|cccc}\toprule
    &  & \text{for } \tida \neq \tid & & \text{for } \loca \neq \loc \\
    & \bmT_\HB(\tid) \eqassn & \bmT_\HB(\tida) \eqassn & \bmW_\HB(\loc)\eqassn  &  \bmW_\HB(\loca)\eqassn\\\midrule
\textbf{Read}
    & \bmT_\HB(\tid) \changed{{}\cup \bmW_\HB(\loc)}
    & \bmT_\HB(\tida)
    & \bmW_\HB(\loc)
    & \bmW_\HB(\loca)
    \\[1ex]
\textbf{Write}
    & \bmT_\HB(\tid) \changed{{}\cup \set{\loc}}
    & \bmT_\HB(\tida) \changed{{}\setminus \set{\loc}}
    & \changed{\bmT_\HB(\tid) \cup \set{\loc}}
    & \bmW_\HB(\loca) \changed{{}\setminus \set{\loc}}
    \\\bottomrule
\end{array}$}
\\[0.5ex]
\resizebox{0.75\width}{!}{$
\begin{array}{c|cccccc}\toprule
    & & \text{for } \tida \neq \tid & & \text{for } \loca \neq \loc
    & & \text{for } \loca \neq \loc \\
    & \bmT_\SC(\tid) \eqassn & \bmT_\SC(\tida) \eqassn & \bmW_\SC(\loc) \eqassn & \bmW_\SC(\loca) \eqassn
    & \bmM_\SC(\loc) \eqassn & \bmM_\SC(\loca) \eqassn \\\midrule
\textbf{Read}
    & \bmT_\SC(\tid) \changed{{}\cup \bmW_\SC(\loc)}
    & \bmT_\SC(\tida)
    & \bmW_\SC(\loc)
    & \bmW_\SC(\loca)
    & \bmM_\SC(\loc) \changed{{}\cup \bmT_\SC(\tid)}
    & \bmM_\SC(\loca)
    \\[1ex]
\textbf{Write}
    & \bmT_\SC(\tid) \changed{{}\cup \bmM_\SC(\loc)}
    & \bmT_\SC(\tida) \changed{{}\setminus \set{\loc}}
    & \changed{\bmT_\SC(\tid) \cup \bmM_\SC(\loc)}
    & \bmW_\SC(\loca) \changed{{}\setminus \set{\loc}}
    & \bmM_\SC(\loc) \changed{{}\cup \bmT_\SC(\tid)}
    & \bmM_\SC(\loca) \changed{{}\setminus \set{\loc}} \\\bottomrule
\end{array}$}
\caption{Maintaining BMs when thread $\tid$ accesses location $\loc$}%
\label{fig:bm:transitions}
\end{figure}

\paragraph{Robustness Verification}
All in all,~\cite{pldi19} established the following result:

\begin{theorem}%
\label{thm:bm:robust}
A program $\prog$ is robust iff running $\prog$ under \SC
with the instrumentation given in \cref{fig:bm:transitions}
never reaches a state in which thread $\tid$ is about to access some location
$\loc \in \bmT_\SC(\tid) \setminus \bmT_\HB(\tid)$.
\end{theorem}

In this theorem, running under \SC means
using the standard interleaving-based operational semantics
with the memory state represented as a mapping from variables to values.

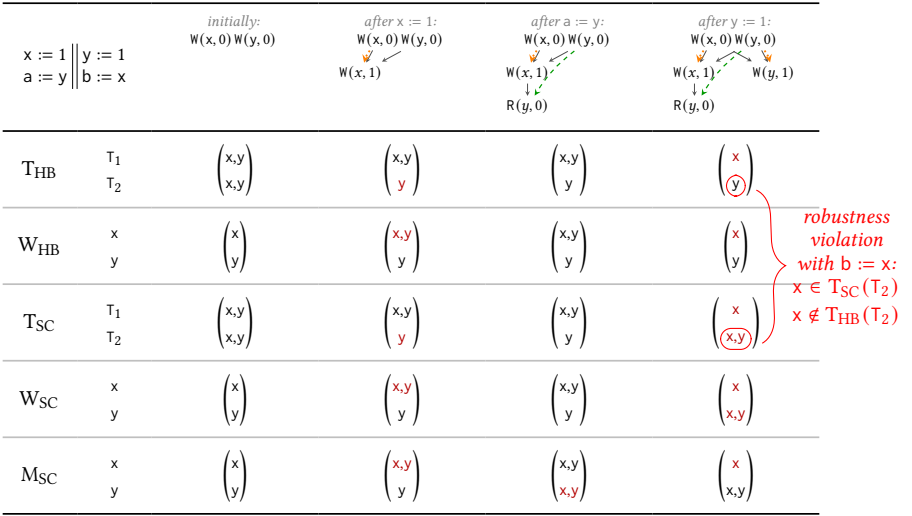
\begin{figure}
\centering\footnotesize
\(\begin{tblr}{rowhead=1, rows={3em,m}, cells={c}, hline{1,2,Z} = {0.08em, black}, hline{3-6} = {0.08em, lightgray},}%
\SetCell[c=2]{l}{{{\resizebox{1.5cm}{!}{\usebox{\sbex}}}}}   &      &	\resizebox{1.8cm}{!}{\usebox{\initb}}           &	\resizebox{1.8cm}{!}{\usebox{\wx}}                 &	\resizebox{1.8cm}{!}{\usebox{\wxry}}                      &	\resizebox{1.8cm}{!}{\usebox{\wxrywy}}      \\
\bmT_\HB & \exl{\ctid1}{\ctid2}     &\exmx{\cloc1,\cloc2}{\cloc1,\cloc2} &	\exmx{\cloc1,\cloc2}{\changed{\cloc2}} &	\exmx{\cloc1,\cloc2}{\cloc2}                  &	\exmx{\changed{\cloc1}}{\circled{\scriptstyle \cloc2}}\tikzmark{hb}          	\\
\bmW_\HB & \exl{\cloc1}{\cloc2}     &\exmx{\cloc1}{\cloc2} &	\exmx{\changed{\cloc1,\cloc2}}{\cloc2} &	\exmx{\cloc1,\cloc2}{\cloc2}                  &	\exmx{\changed{\cloc1}}{\cloc2}                   	\\

\bmT_\SC & \exl{\ctid1}{\ctid2}     &\exmx{\cloc1,\cloc2}{\cloc1,\cloc2} &	\exmx{\cloc1,\cloc2}{\changed{\cloc2}} &	\exmx{\cloc1,\cloc2}{\cloc2}                  &	\exmx{\changed{\cloc1}}{\circled{\changed{\scriptstyle \cloc1,\cloc2}}}\tikzmark{sc}  \\
\bmW_\SC & \exl{\cloc1}{\cloc2}     &\exmx{\cloc1}{\cloc2} &	\exmx{\changed{\cloc1,\cloc2}}{\cloc2} &	\exmx{\cloc1,\cloc2}{\cloc2}                  &	\exmx{\changed{\cloc1}}{\changed{\cloc1,\cloc2}} \\
\bmM_\SC & \exl{\cloc1}{\cloc2}     &\exmx{\cloc1}{\cloc2} &	\exmx{\changed{\cloc1,\cloc2}}{\cloc2} &	\exmx{\cloc1,\cloc2}{\changed{\cloc1,\cloc2}} & \exmx{\changed{\cloc1}}{\cloc1,\cloc2}	\\
\end{tblr}\)
\caption{BM-instrumentation maintenance for a run of \ref{prog:sb}}%
\label{fig:sb_bm}
\begin{tikzpicture}[remember picture,overlay,font=\footnotesize]
  \draw [red,decorate,decoration={brace,amplitude=10pt,raise=-1pt},xshift=-0.5cm,yshift=0pt]
  ($(pic cs:hb)+(0,-.2)$) -- ($(pic cs:sc)+(0,-.2)$) node [text width=2cm,black,midway,xshift=1.4cm]
  {\emph{\color{red}{\shortstack{robustness\\violation\\with $\assignInst{\creg2}{\cloc1}$:\\$\cloc1 \in \bmT_\SC(\ctid2)$\\$\cloc1\nin\bmT_\HB(\ctid2)$}}}};
\end{tikzpicture}
\end{figure}

\begin{example}%
\label{ex:bm_sb}
Consider again~\ref{prog:sb} from \cref{fig:intro_rob},
and its run described in \cref{ex:SB_witness}.
The BM instrumentation after each instruction is depicted in \cref{fig:sb_bm}.
When $\ctid1$ executes $\assignInst{\cloc1}{1}$, we remove $\cloc1$ from views associated with $\ctid2$.
Then, $\ctid1$ reads $0$ from $\cloc2$ and we add $\cloc1$ to
$\bmM_\SC(\cloc2)$ since this read of $\cloc2$ is $\lHBSC$-after the last write to $\cloc1$.
When $\ctid2$ then executes $\assignInst{\cloc2}{1}$,
it gains the `knowledge' from $\cloc2$, and we add $\cloc1$ to $\bmT_\SC(\ctid2)$.
Finally, when $\ctid2$ tries to read from $\cloc1$ we flag a robustness violation as:
$\cloc1 \in \bmT_\SC(\ctid2) \setminus \bmT_\HB(\ctid2)$.
\end{example}

\section{Location-Clock-Based Robustness Analysis}%
\label{sec:algorithm}

\newcommand{\figvc}{
\begin{figure}%
\centering
\resizebox{0.75\width}{!}{$
\begin{array}{c|cccc}\toprule
    & & \text{for } \tida \neq \tid & & \text{for } \loca \neq \loc \\
    & \vcT_\HB(\tid) \eqassn & \vcT_\HB(\tida) \eqassn & \vcW_\HB(\loc) \eqassn & \vcW_\HB(\loca) \eqassn \\\midrule
\textbf{Read}
    & \vcT_\HB(\tid) \changed{{}\sqcup \vcW_\HB(\loc)}
    & \vcT_\HB(\tida)
    & \vcW_\HB(\loc)
    & \vcW_\HB(\loca)
    \\[1ex]
\textbf{Write}
    & \vcT_\HB(\tid) \changed{{}\sqcup (\epoch{\loc}{\ts})}
    & \vcT_\HB(\tida)
    & \changed{\vcT_\HB(\tid) \sqcup (\epoch{\loc}{\ts})}
    & \vcW_\HB(\loca)
    \\\bottomrule
\end{array}$}
\\[2ex]
\resizebox{0.75\width}{!}{$
\begin{array}{c|cccccc}\toprule
    & & \text{for } \tida \neq \tid & & \text{for } \loca \neq \loc
    & & \text{for } \loca \neq \loc \\
    & \vcT_\SC(\tid) \eqassn & \vcT_\SC(\tida) \eqassn & \vcW_\SC(\loc) \eqassn & \vcW_\SC(\loca) \eqassn
    & \vcM_\SC(\loc) \eqassn & \vcM_\SC(\loca) \eqassn \\\midrule
\textbf{Read}
    & \vcT_\SC(\tid) \changed{{}\sqcup \vcW_\SC(\loc)}
    & \vcT_\SC(\tida)
    & \vcW_\SC(\loc)
    & \vcW_\SC(\loca)
    & \vcM_\SC(\loc) \changed{{}\sqcup \vcT_\SC(\tid)}
    & \vcM_\SC(\loca)
    \\[1ex]
\textbf{Write}
    & \vcT_\SC(\tid) \changed{{}\sqcup \vcM_\SC(\loc) \sqcup (\epoch{\loc}{\ts})}
    & \vcT_\SC(\tida)
    & \changed{\vcT_\SC(\tid) \sqcup \vcM_\SC(\loc) \sqcup (\epoch{\loc}{\ts})}
    & \vcW_\SC(\loca)
    & \vcM_\SC(\loc) \changed{{}\sqcup \vcT_\SC(\tid) \sqcup (\epoch{\loc}{\ts})}
    & \vcM_\SC(\loca)
    \\\bottomrule
\end{array}$}
\caption{Maintaining LCs when thread $\tid$ accesses location $\loc$,
where $\ts = {\vcW_\HB}(\loc)(\loc) +1= {\vcW_\SC}(\loc)(\loc) +1$ }%
\label{fig:vc:transitions_ra}
\end{figure}
}

The instrumentation in \cref{fig:bm:transitions}, previously utilized for static robustness verification,
can also be directly used in a dynamic analysis.
However, this approach presents significant limitations, as detailed in \cref{sec:BM_bad}.
Instead, we introduce a novel instrumentation for robustness in \cref{sec:LC},
which is better suited for dynamic analysis.

\subsection{Drawbacks of BMs for Dynamic Analysis}%
\label{sec:BM_bad}

The BM instrumentation was specifically developed for static verification,
and, as is, it is not suitable for a dynamic analysis.
This stems from two reasons.

\paragraph{Blocking Instrumentation}
For static verification, it is essential that the instrumentation remains constant in size
(assuming a fixed number of threads and memory locations).
Constant-size instrumentation, like the BM instrumentation, is necessary for reducing the robustness verification problem
to a reachability problem in a finite-state machine.
In turn, for a dynamic approach, a constant-size instrumentation is unnecessary
and different requirements come into play.
A primary consideration in this setting is the extent to which concurrency is compromised
in the instrumented program compared to the uninstrumented one.
In this regard, BMs pose significant drawbacks:
any write operation %
necessitates updating the
BMs across \emph{all threads and locations} (see \cref{fig:bm:transitions}),
as they all lose ``awareness'' of the last write.
This update must occur atomically, thereby blocking concurrent threads from writing simultaneously.
Consequently, maintaining the BMs effectively serializes all memory writes, severely limiting parallelism in the original program.

\paragraph{Non-Predictive Instrumentation}
For static verification, there is no need to \emph{predict} the existence
of a non-robustness witness from an execution graph that itself is not a witness,
as we anyway exhaustively account for all possible executions.
However, dynamic analysis, by nature, observes specific executions,
and the inability to predict a non-robustness witness is a significant drawback.
Without predictive capabilities,
the analysis might miss non-robustness witnesses that rarely appear in the observed executions.
The next example demonstrates this point.

\begin{example}%
\label{ex:sb_writes}
Consider the following variant of~\ref{prog:sb}:

\begin{minipage}{0.3\textwidth}
\begin{equation*}
\inarrII{\assignInst{\cloc1}{1} \\ \assignInst{\creg1}{\cloc2} \\ \assignInst{\cloc1}{2}}
{\assignInst{\cloc2}{1} \\ \assignInst{\creg2}{\cloc1} \\ \assignInst{\cloc2}{2}}
\end{equation*}
\end{minipage}
\hfill
\begin{minipage}{0.6\textwidth}
This program is not robust: as with~\ref{prog:sb},
it is possible under weak memory for both threads to read the initial value of $0$,
an outcome not permitted under \SC.
\end{minipage}

\smallskip
\noindent
The BM-based algorithm is complete and, indeed, it can detect a non-robustness witness in certain runs.
For instance, in a run where $\ctid1$ executes its first two instructions
and then $\ctid2$ follows, just before $\assignInst{\creg2}{\cloc1}$,
we observe that $\cloc1 \in \bmT_\SC(\ctid2) \setminus \bmT_\HB(\ctid2)$
and flag a robustness violation.

However, consider a run in which thread $\ctid1$ terminates before $\ctid2$ begins.
In this sequence, the last write to $\cloc1$ when $\ctid2$ executes
is the event from $\assignInst{\cloc1}{2}$.
Since there is no $\lHBSC$ path connecting this write event to thread $\ctid2$,
we find $\cloc1 \nin \bmT_\SC(\ctid2)$,
meaning that BM does not detect a violation in this run.
This reveals a limitation of \cref{thm:bm:robust} for dynamic verification:
even when a prefix of the execution graph witnesses a robustness violation, the current run itself does not flag this issue.

In fact, running a BM-based dynamic analysis on this program over 1,000 executions did not detect the violation.
This is because the violation is captured in the BM-based analysis only with a context switch after one thread's first two instructions. Given the short program length, such a context switch is rare without explicitly adjusting the OS scheduling.
\end{example}

\subsection{Location Clocks to the Rescue}%
\label{sec:LC}

To overcome the limitations of BMs, we introduce \emph{location clocks},
a variant of vector clocks, and show how to adapt the instrumentation to use location clocks.
Location clocks (LCs) have two primary advantages over over BMs.
First, LCs are more efficient when maintaining instrumentation,
as a thread accessing a location only needs to update the location clocks associated with
that specific thread and location.
This is critical for dynamic analysis,
as it enables locking only the instrumentation relevant to the accessed location,
and such fine-grained locking allows for significantly higher concurrency during the (instrumented) program execution.
Second, LCs support the \emph{prediction} of non-robustness witnesses,
effectively addressing the challenge illustrated in \cref{ex:sb_writes}.
We now proceed with the formal definitions.

\paragraph{Location Clocks}
Let $\Time$ denote the set of natural numbers.
We refer to the elements of $\Time$ as \emph{timestamps} and use $\ts$ to range over them.
A \emph{location epoch} is a pair $\epoch{\loc}{\ts}$
where $\loc\in\Loc$ and $\ts\in\Time$, and a
\emph{location clock} (or LC) $\vc$ is a set of location epochs such that
each location $\loc$ appears in at most one location epoch in $\vc$.
We identify a location epoch $\epoch{\loc}{\ts}$ with a singleton LC $\set{\epoch{\loc}{\ts}}$.
An LC $\vc$ essentially represents a function from $\Loc$ to $\Time$,
assigning the (unique) timestamp $\ts$ such that $\epoch{\loc}{\ts}\in \vc$ for every $\loc\in\Loc$
that appears in $\vc$, and $0$ for every other $\loc\in\Loc$.
We often identify LCs with the functions they represent,
writing, \eg~$\vc(\loc)$ for the timestamp $\vc$ assigns to $\loc$.
The empty LC, representing the function $\lambda \loc \ldotp 0$, is denoted by $\vc_\Init$,
and we let $\vc_1 \sqcup \vc_2 \defeq \lambda \loc \ldotp \max \set{\vc_1(\loc), \vc_2(\loc)}$.

\paragraph{Location-Clock Instrumentation for Robustness}
LCs can replace BMs for robustness verification.
Intuitively, each write to a location $\loc$ obtains a monotonically increasing timestamp,
and we keep track of the last timestamp for each location that each thread ``knows about'' in the $\lHB$ and $\lHBSC$ relations.
Formally, given an execution graph $G$, the modification order $G.\lMO$ is interpreted as
a timestamp mapping $G.\lTS: G.\sW \to \Time$ defined by
$G.\lTS(w) \defeq \size{\set{w'\in G.\sW \st \tup{w',w}\in G.\lMO}}$,
\ie the number of write events that are $\lMO$-before $w$.
The instrumentation consists of mappings:
$\vcT_\HB$ and $\vcT_\SC$ that assign an LC to every thread;
and $\vcW_\HB$, $\vcW_\SC$, and $\vcM_\SC$ that assign an LC to every location.
The meaning of these LCs for a given execution graph $G$ is summarized as follows:
\begin{align*}
\vcT_\HB(\tid)(\loc) & =\max\set{G.\lTS(w) \st w \in \sW_\loc \land \exists e \in G.\lE^\tid \ldotp \tup{w,e} \in G.\lHB^?}
\\
\vcW_\HB(\loca)(\loc) & =\max\set{G.\lTS(w) \st w \in \sW_\loc \land \tup{w,G.\wmax{\loca}}\in G.\lHB^?}
\\
\vcT_\SC(\tid)(\loc) & =\max\set{G.\lTS(w) \st w \in \sW_\loc \land \exists e \in G.\lE^\tid \ldotp \tup{w,e} \in G.\lHBSC^?}
\\
\vcW_\SC(\loca)(\loc) & =\max\set{G.\lTS(w) \st w \in \sW_\loc \land \tup{w,G.\wmax{\loca}}\in G.\lHBSC^?}
\\
\vcM_\SC(\loca)(\loc) & =\max\set{G.\lTS(w) \st w \in \sW_\loc \land \exists e \in G.\lE_\loca \ldotp \tup{w,e} \in G.\lHBSC^?}
\end{align*}
Above we take $\max\emptyset \defeq 0$.
Note that for every $\loc\in\Loc$, $\vcW_\HB(\loc)(\loc)=\vcW_\SC(\loc)(\loc)$,
both storing the maximal timestamp among all writes to $\loc$
(which is $\size{G.\sW_\loc}-1$).

\figvc

Initially, all LCs are initialized to $\vc_\Init$.
The maintenance of the instrumentation with each access to a shared variable is given in \cref{fig:vc:transitions_ra}.
The rows correspond to operations performed by the program,
while each column describes how the instrumentation is updated.
For example, when thread $\tid$ reads from location $\loc$,
$\tid$ incorporates $\loc$'s view into its own view, thus ``learning'' whatever $\loc$ ``knew''.
When $\tid$ writes to location $\loc$, the timestamp for $\loc$ is incremented,
and added to ${\vcT_\HB}(\tid)$.
In addition, $\tid$ ``releases its knowledge'',
sharing it with the location so that any thread that will read $\loc$ in the future will be able to absorb it.
As a side effect, due to increment of the timestamp,
all other locations no longer have the maximal timestamp for $\loc$ and thus are unaware of the last write to $\loc$.
The \SC transitions follow the same intuitions.
In particular, to track $\lFR$, read events propagate the thread view $\vcT_\SC$ to $\vcM_\SC$,
which is acquired by later writes to the same location.

To test for robustness violation, when a thread accesses $\loc$,
we compare the maximal timestamp for $\loc$ the thread ``knows about''
in $\lHBSC$ and in $\lHB${}\@.
If the first is strictly greater, we flag a violation.

\begin{theorem}%
\label{thm:main}
A program $\prog$ is robust iff running $\prog$ under \SC with the instrumentation given in \cref{fig:vc:transitions_ra}
never reaches a state in which thread $\tid$ is about to access location $\loc$ but
	${\vcT_\HB}(\tid)(\loc) < {\vcT_\SC}(\tid)(\loc)$.
\end{theorem}
\begin{proof}[Proof (sketch)]
We first show that the LC instrumentation indeed tracks the timestamps as described above.
Then, if $\prog$ is not robust, then by~\cite[Thm. 5.1.]{pldi19}, there exists a non-robustness witness for $\prog$,
which implies that under \SC $\prog$ generates some (\SC-consistent) execution graph $G$
with $\tup{G.\wmax{\loc},e}\in G.\lHBSC$ for some $e\in G.\lE^\tid$ but $\tup{G.\wmax{\loc},e}\nin G.\lHB$
for every $e\in G.\lE^\tid$, and the next step of thread $\tid$ is an access to location $\loc$.
In that run, we will have ${\vcT_\HB}(\tid)(\loc) < {\vcT_\SC}(\tid)(\loc) = G.\lTS(G.\wmax{\loc})$.

For the converse, we first note that a run of $\prog$ under \SC that reaches
a state in which thread $\tid$ is about to access location $\loc$ but
	${\vcT_\HB}(\tid)(\loc) < {\vcT_\SC}(\tid)(\loc)$	implies the existence of an 
	\emph{extended non-robustness witness}.
	The latter is defined as a configuration in which some thread $\tid$ is about to access a location $\loc$
and the following hold for the (\SC-consistent) generated execution graph $G$ and some $w \in G.\sW_\loc$:
\begin{enumerate}
\item The next access of $\tid$ races with $w$, as well as with every write $w'$ that is $G.\lMO$-after $w$, 
\ie $\tup{w,e} \nin G.\lMO^? \seq G.\lHB$ for every $e \in \sE^\tid$.
\item The next access of $\tid$ cannot have been executed before $w$,
\ie  $\tup{w,e} \in G.\lHBSC$ for some $e \in \sE^\tid$.
\end{enumerate}
Now, the existence of an extended non-robustness witness implies a standard non-robustness witness,
from which by~\cite[Thm. 5.1.]{pldi19}, it follows that $\prog$ is not robust.
Indeed, let $w_0$ be the $G.\lMO$-maximal such write event,
and construct an execution graph $G'$ by restricting $G$ to all nodes $e'$ such that
$\tup{e',e}\in G.\lHBSC$. It is easy to see that there exists another run of $\prog$
under \SC that reaches a state in which thread $\tid$ is about to access location $\loc$,
and the generated graph in this run is $G'$.
Moreover, we have $G'.\wmax{\loc} = w_0$, and so
$\tup{G'.\wmax{\loc},e}\in G'.\lHBSC$ for some $e\in G'.\lE^\tid$
but $\tup{G'.\wmax{\loc},e}\nin G'.\lHB$ for every $e\in G'.\lE^\tid$.
\end{proof}

\Cref{thm:main} provides the formal basis for the guarantees of the dynamic analysis.
The '$\Rightarrow$' direction ensures soundness, meaning no spurious violations.
The converse ensures completeness,
indicating that for a non-robust program, some schedule exposes the violation.

\newcommand{\exvc}[2]{%
\begin{adjblockarray}{c}
\begin{adjblockarray}{\{c\}}
  {\scriptstyle #1}
\end{adjblockarray}\\
\begin{adjblockarray}{\{c\}}
  {\scriptstyle #2}
\end{adjblockarray}
\end{adjblockarray}
}

 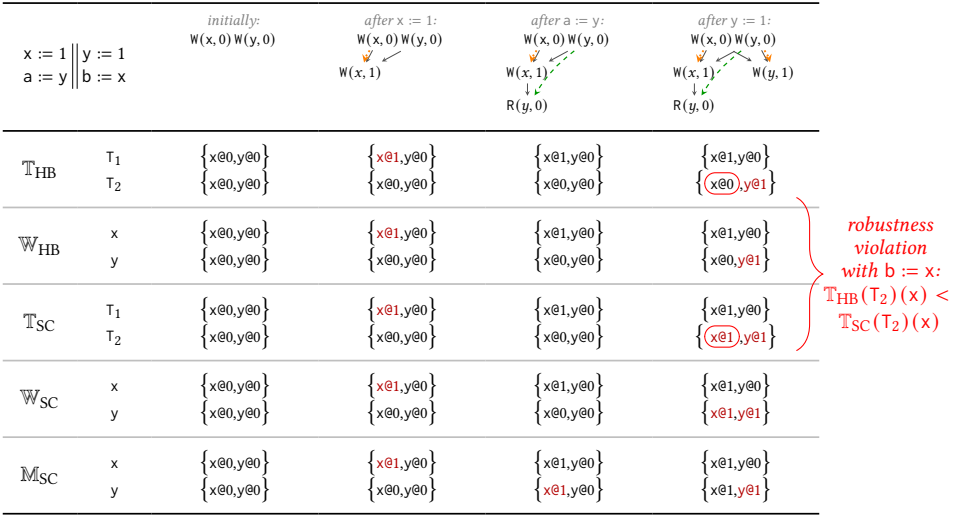
\begin{figure}
\centering\footnotesize
\(\begin{tblr}{rowhead=1, rows={3em,m}, cells={c}, hline{1,2,Z} = {0.08em, black}, hline{3-6} = {0.08em, lightgray},}%
\SetCell[c=2]{l}{{{\resizebox{1.5cm}{!}{\usebox{\sbex}}}}}   &      &	\resizebox{1.8cm}{!}{\usebox{\initb}}           &	\resizebox{1.8cm}{!}{\usebox{\wx}}                 &	\resizebox{1.8cm}{!}{\usebox{\wxry}}                      &	\resizebox{1.8cm}{!}{\usebox{\wxrywy}}      \\
\vcT_\HB & \exl{\ctid1}{\ctid2} &\exvc{\epoch{\cloc1}{0},\epoch{\cloc2}{0}}{\epoch{\cloc1}{0},\epoch{\cloc2}{0}} &	\exvc{\changed{\epoch{\cloc1}{1}},\epoch{\cloc2}{0}}{\epoch{\cloc1}{0},\epoch{\cloc2}{0}} &	\exvc{\epoch{\cloc1}{1},\epoch{\cloc2}{0}}{\epoch{\cloc1}{0},\epoch{\cloc2}{0}} &	   \exvc{\epoch{\cloc1}{1},\epoch{\cloc2}{0}}{\circled{\scriptstyle\epoch{\cloc1}{0}},\changed{\epoch{\cloc2}{1}}}\tikzmark{vhb}          	\\
\vcW_\HB & \exl{\cloc1}{\cloc2} &\exvc{\epoch{\cloc1}{0},\epoch{\cloc2}{0}}{\epoch{\cloc1}{0},\epoch{\cloc2}{0}} &	\exvc{\changed{\epoch{\cloc1}{1}},\epoch{\cloc2}{0}}{\epoch{\cloc1}{0},\epoch{\cloc2}{0}} &	\exvc{\epoch{\cloc1}{1},\epoch{\cloc2}{0}}{\epoch{\cloc1}{0},\epoch{\cloc2}{0}} &	   \exvc{\epoch{\cloc1}{1},\epoch{\cloc2}{0}}{\epoch{\cloc1}{0},\changed{\epoch{\cloc2}{1}}} 	                       	\\

\vcT_\SC & \exl{\ctid1}{\ctid2} &\exvc{\epoch{\cloc1}{0},\epoch{\cloc2}{0}}{\epoch{\cloc1}{0},\epoch{\cloc2}{0}} &	\exvc{\changed{\epoch{\cloc1}{1}},\epoch{\cloc2}{0}}{\epoch{\cloc1}{0},\epoch{\cloc2}{0}} &	\exvc{\epoch{\cloc1}{1},\epoch{\cloc2}{0}}{\epoch{\cloc1}{0},\epoch{\cloc2}{0}} &	   \exvc{\epoch{\cloc1}{1},\epoch{\cloc2}{0}}{\circled{\changed{\scriptstyle\epoch{\cloc1}{1}}},\changed{\epoch{\cloc2}{1}}}\tikzmark{vsc}  \\
\vcW_\SC & \exl{\cloc1}{\cloc2} &\exvc{\epoch{\cloc1}{0},\epoch{\cloc2}{0}}{\epoch{\cloc1}{0},\epoch{\cloc2}{0}} &	\exvc{\changed{\epoch{\cloc1}{1}},\epoch{\cloc2}{0}}{\epoch{\cloc1}{0},\epoch{\cloc2}{0}} &	\exvc{\epoch{\cloc1}{1},\epoch{\cloc2}{0}}{\epoch{\cloc1}{0},\epoch{\cloc2}{0}} &	   \exvc{\epoch{\cloc1}{1},\epoch{\cloc2}{0}}{\changed{\epoch{\cloc1}{1}},\changed{\epoch{\cloc2}{1}}} \\
\vcM_\SC & \exl{\cloc1}{\cloc2} &\exvc{\epoch{\cloc1}{0},\epoch{\cloc2}{0}}{\epoch{\cloc1}{0},\epoch{\cloc2}{0}} &	\exvc{\changed{\epoch{\cloc1}{1}},\epoch{\cloc2}{0}}{\epoch{\cloc1}{0},\epoch{\cloc2}{0}} &	\exvc{\epoch{\cloc1}{1},\epoch{\cloc2}{0}}{\changed{\epoch{\cloc1}{1}},\epoch{\cloc2}{0}} &\exvc{\epoch{\cloc1}{1},\epoch{\cloc2}{0}}{\epoch{\cloc1}{1},\changed{\epoch{\cloc2}{1}}}	\\
\end{tblr}\)
 \caption{LC-instrumentation maintenance for a run of \ref{prog:sb}}%
 \label{fig:sb_vc}
\begin{tikzpicture}[remember picture,overlay,font=\footnotesize]
  \draw [red,decorate,decoration={brace,amplitude=10pt,raise=1pt},xshift=-0.5cm,yshift=0pt]
  ($(pic cs:vhb)+(0,-.3)$) -- ($(pic cs:vsc)+(0,-.3)$) node [text width=2cm,black,midway,xshift=1.4cm]
        {\emph{\color{red}{\shortstack{robustness\\violation\\with $\assignInst{\creg2}{\cloc1}$:\\$\vcT_\HB(\ctid2)(\cloc1) <{}$\\$\vcT_\SC(\ctid2)(\cloc1)$}}}};
\end{tikzpicture}
 \end{figure}

\begin{example}\label{ex:vc_sb}
Consider again~\ref{prog:sb} from \cref{fig:intro_rob},
and its run described in \cref{ex:SB_witness}.
The LC instrumentation after each instruction is depicted in \cref{fig:sb_vc}.
When $\ctid1$ executes $\assignInst{\cloc1}{1}$, we increment $\cloc1$ timestamp
and update the LCs for $\ctid1$ and $\cloc1$.
Then, $\ctid1$ reads $0$ for $\cloc2$ and we update the timestamp for  $\cloc1$ in
$\vcM_\SC(\cloc2)$ since the read of $\cloc2$ is aware of the last write to $\cloc1$.
When $\ctid2$ then executes $\assignInst{\cloc2}{1}$,
it learns about the events from $\cloc2$, specifically the last write to $\cloc1$.
Finally, when $\ctid2$ tries to read from $\cloc1$ we flag a robustness violation as:
${\vcT_\HB}(\ctid2)(\cloc1) < {\vcT_\SC}(\ctid2)(\cloc1)$.
To see the advantage over BMs compare with \cref{ex:bm_sb}.
For instance, when $\ctid2$ executes $\assignInst{\cloc2}{1}$,
we only update the LCs associated with $\ctid2$ and $\cloc2$, whereas in the BM instrumentation
we needed to update almost all BMs.
\end{example}

We note that the space required for the LC instrumentation is
$O((\size{\Tid}+\size{\Loc})\cdot\size{\Loc}\cdot \log n)$,
where $n$ represents the number of write events in the analyzed run.
The $\log n$ factor arises from the need to record timestamps.
In practice, however, timestamps are bounded and stored as large integers ($64$ bits),
making the actual memory consumption independent of the length of the run.

\paragraph{Predictive Analysis Using LCs}
We highlight the predictive advantage of \cref{thm:main}
over \cref{thm:bm:robust} from~\cite{pldi19}.
First, revisiting \cref{ex:sb_writes}, note that if we run
$\ctid1$ to completion before $\ctid2$ begins, 
we still have that ${\vcT_\HB}(\ctid2)(\cloc1) < {\vcT_\SC}(\ctid2)(\cloc1)$
when $\ctid2$ tries to read from $\cloc1$.
Indeed, the final write by $\ctid1$ to $\cloc1$ does not change 
the LCs associated with $\ctid2$.
As a result, the LC-based analysis based on \cref{thm:main}
detects the violation in \emph{any run} of the program (as confirmed in our experiments).

More generally, LCs are able to detect \emph{extended} non-robustness witnesses,
as defined in the proof of \cref{thm:main}, which imply the existence of a non-robustness witness in a different run. 
One can use LCs to detect standard witnesses (without prediction), by 
a straightforward adaptation of \cref{thm:bm:robust} to LCs, based on the following 
relationship between the LC and BM instrumentations:
\[
	{\vcT_\HB}(\tid)(\loc) = {\vcW_\HB}(\loc)(\loc) \Leftrightarrow \loc\in\bmT_\HB(\tid) \qquad\qquad 
	 {\vcT_\SC}(\tid)(\loc) = {\vcW_\SC}(\loc)(\loc) \Leftrightarrow \loc\in\bmT_\SC(\tid)
\]
Thus, to detect standard witnesses we should flag a violation
when thread $\tid$ is about to access location $\loc$ and 
${\vcT_\HB}(\tid)(\loc) <  {\vcW_\HB}(\loc)(\loc) = {\vcW_\SC}(\loc)(\loc) = {\vcT_\SC}(\tid)(\loc)$,
which means that $\tid$ is aware of the globally maximal write to $\loc$ in $\lHBSC$, yet
unaware of the globally maximal write to $\loc$ in $\lHB$.
By contrast, the (weaker) condition ${\vcT_\HB}(\tid)(\loc) < {\vcT_\SC}(\tid)(\loc)$ in \cref{thm:main}
means that the maximal write to $\loc$ that $\tid$ is $\lHB$-aware of
is $\lMO$-before the maximal write to $\loc$ that $\tid$ is $\lHBSC$-aware of.
The latter condition allows the LC-based analysis to be predictive,
which enhances the non-robustness detection rate.

\paragraph{Relation to Race Detection}
We discuss the relation to standard dynamic race-detection that uses vector clocks
(see \cite{Mansky17} for a formal description).
We observe that extended non-robustness witnesses are races on atomic accesses,
but, of course, not every such race imply non-robustness.
Besides identifying a race on location $\loc$ between an event $e_1$ and a later executed event $e_2$
(\ie two accesses accessing $\loc$ such that at least one of them is a write,
and $\tup{e_1,e_2} \nin G.\lHB \seq G.\lPO$), 
in extended non-robustness witnesses one also requires that:
\begin{enumerate*}[label=(\roman*)]
\item $e_1$ has to be a write event
(so we do not check for so-called ``read-write'' races);
\item every write event to $\loc$ executed after $e_1$ races with $e_2$
(\ie $\tup{e_1,e_2} \nin G.\lMO \seq G.\lHB \seq G.\lPO$);
and \item $e_2$ must be executed after $e_1$ in any run under \SC
(\ie $\tup{e_1,e_2} \in G.\lHBSC \seq G.\lPO$).
\end{enumerate*}
Tracking such patterns in a dynamic analysis requires non-trivial extensions of standard race detection.
Our technique does so by switching from standard vector clocks 
(functions from $\Tid$ to $\Time$) to location clocks (functions from $\Loc$ to $\Time$).

\section{Key Enhancements to the Basic Approach}%
\label{sec:extensions}

We introduce several essential extensions to the basic approach, 
broadening its applicability to real-world concurrent programs. 
We discuss support for read-modify-write instructions (\cref{sec:rmws}),
extension to the full \RLX mode (\cref{sec:rc20}),
and a technique for masking benign robustness violations caused by busy loops (\cref{sec:wait}).
This discussion remains informal, focusing on the main concepts behind each extension. 
Details integrating all extensions are provided in
\iflong
\cref{app:rlx}.
\else
\cite[§A]{appendix}.
\fi

\subsection{Supporting Read-Modify-Writes}%
\label{sec:rmws}

Our discussion so far has focused on reads and writes, 
omitting atomic read-modify-write (RMW) instructions, 
such as fetch-and-add (FADD) and compare-and-swap (CAS)\footnote{The extension discussed
in this section applies to \emph{weak} CAS, 
which may fail even when it reads the expected value. 
To support \emph{strong} CAS, this extension must be combined with the value tracking discussed in~\cref{sec:wait}. 
Our full development, along with the accompanying tool, supports both types of CAS.}, 
which are indispensable in concurrent programs. 
Our method extends to support RMWs.

To model RMWs in execution graphs, C11 uses RMW events, 
labeled with $\ulab{}{\loc}{\val_\lR}{\val_\lW}$, 
denoting the location $\loc$ being updated, 
the value $\val_\lR$ read from $\loc$, 
and the value $\val_\lW$ that replaces $\val_\lR$. 
RMW events behave both like reads (\eg~require an incoming $\lRF$-edge to justify the value being read) 
and writes (\eg~totally ordered per location with other writes via $\lMO$). 
The distinctiveness of RMWs, beyond simply combining a read and a write, 
is captured by another constraint in consistent graphs:\footnote{With $\lU$ events, the definition of $G.\lFR$
has to exclude the identity relation:  $G.\lFR \defeq (G.\lRF^{-1}\seq G.\lMO) \setminus \set{\tup{e,e} \st e\in \sE}$.}

\begin{itemize}
\item $G.\lFR\seq G.\lMO$ is irreflexive.\ \labelAxiom{atomicity}{ax:at}
\end{itemize}
This condition requires that the oncoming $G.\lRF$-edge to each RMW event $e$ 
originates from the \emph{immediate} $G.\lMO$-predecessor of $e$.

Due to the~\ref{ax:at}, supporting RMWs in robustness analysis requires specific adjustments. 
Indeed, without RMWs, if the next access by thread $\tid$ is a write to location $\loc$, 
and some write $w$ to $\loc$ has no $G.\lMO^? \seq G.\lHB^?$ path to thread $\tid$, 
we can always place the next write as the immediate predecessor of $w$ in $G.\lMO$. 
This would indicate a robustness violation 
when $w$ has $G.\lHBSC$ to thread $\tid$, and thus no SC-consistent execution allows this placement.
However, with RMWs,~\ref{ax:at} forbids the next write from being placed immediately 
before $w$ in $G.\lMO$ when $w$ itself is an RMW\@.

To handle RMWs, our approach extends the LC instrumentation by incorporating an additional timestamp at each location, 
tracking the cumulative count of \emph{non-RMW writes} performed for the location.
Each thread and location maintains LCs containing these timestamps as well, 
reflecting their awareness of this timestamp in $\lHB$ and $\lHBSC$.
Upon each write (or RMW), we use these LCs to detect potential robustness violations
(whereas upon reads we still use the standard LCs).
Thus, a violation is flagged only if a write can indeed be added to the current graph in a way that 
preserves weak memory consistency, including~\ref{ax:at}, yet disrupts SC-consistency.

\begin{figure}
\begin{minipage}{0.3\textwidth}
\begin{center}
(1)\quad$
\inarrII{\faddInst{\creg1}{\cloc1}{1} \\ \assignInst{\creg2}{\cloc2}}
{\assignInst{\cloc2}{1} \\ \assignInst{\creg3}{\cloc1}}
$
\end{center}

\vspace*{50pt}

\begin{center}
(2)\quad$
\inarrII{\faddInst{\creg1}{\cloc1}{1} \\ \assignInst{\creg2}{\cloc2}}
{\assignInst{\cloc2}{1} \\ \assignInst{\cloc1}{3}}
$
\end{center}
\end{minipage}\hfill
\begin{minipage}{0.35\textwidth}
\begin{center}
$$G_1=\vcenter{\hbox{\begin{tikzpicture}[yscale=0.8,xscale=1.05]
  \node (0x)  at (-0.7,2) {$\evlab{\lW}{}{\cloc1}{0}$};
  \node (0y)  at (0.7,2) {$\evlab{\lW}{}{\cloc2}{0}$};
  \node (11)  at (-1,1) {$\evlab{\lU}{}{\cloc1}{0,1}$ };
  \node (12)  at (-1,0) {$\evlab{\lR}{}{\cloc2}{0}$ };
  \node (21)  at (1,1) {$\evlab{\lW}{}{\cloc2}{1}$ };
    \node (22)  at (1,0) {$\evlab{\lR}{}{\cloc1}{0}$ };
  \draw[po] (11) edge (12);
  \draw[po] (0x) edge (11) edge (21);
  \draw[po] (0y) edge (11) edge (21);
    \draw[mo,bend right=15] (0x) edge (11);
    \draw[rf,bend right=45] (0x) edge (11);
      \draw[rf,bend left=10] (0y) edge (12);
  \draw[mo,bend left=15] (0y) edge (21);  
    \draw[po] (21) edge (22);
      \draw[rf,bend right=10] (0x) edge (22);
\end{tikzpicture}}}$$
\[G_2=\vcenter{\hbox{\begin{tikzpicture}[yscale=0.8,xscale=1.05]
  \node (0x)  at (-0.7,2) {$\evlab{\lW}{}{\cloc1}{0}$};
  \node (0y)  at (0.7,2) {$\evlab{\lW}{}{\cloc2}{0}$};
  \node (11)  at (-1,1) {$\evlab{\lU}{}{\cloc1}{0,1}$ };
  \node (12)  at (-1,0) {$\evlab{\lR}{}{\cloc2}{0}$ };
  \node (21)  at (1,1) {$\evlab{\lW}{}{\cloc2}{1}$ };
      \node (22)  at (1,0) {$\evlab{\lW}{}{\cloc1}{3}$ };
  \draw[po] (11) edge (12);
  \draw[po] (0x) edge (11) edge (21);
  \draw[po] (0y) edge (11) edge (21);
    \draw[mo,bend right=15] (0x) edge (11);
    \draw[rf,bend right=45] (0x) edge (11);
      \draw[rf,bend left=10] (0y) edge (12);
  \draw[mo,bend left=15] (0y) edge (21);
        \draw[mo,bend right=10] (0x) edge (22);
                \draw[mo] (22) edge (11);
\end{tikzpicture}}}\]
\end{center}
\end{minipage}\hfill\vrule\hfill
\begin{minipage}{0.25\textwidth}
\begin{center}
\[G=\vcenter{\hbox{\begin{tikzpicture}[yscale=1,xscale=0.8]
  \node (0x)  at (-0.7,2) {$\evlab{\lW}{}{\cloc1}{0}$};
  \node (0y)  at (0.7,2) {$\evlab{\lW}{}{\cloc2}{0}$};
  \node (11)  at (-1,1) {$\evlab{\lU}{}{\cloc1}{0,1}$ };
  \node (12)  at (-1,0) {$\evlab{\lR}{}{\cloc2}{0}$ };
  \node (21)  at (1,1) {$\evlab{\lW}{}{\cloc2}{1}$ };
  \draw[po] (11) edge (12);
  \draw[po] (0x) edge (11) edge (21);
  \draw[po] (0y) edge (11) edge (21);
    \draw[mo,bend right=15] (0x) edge (11);
    \draw[rf,bend right=45] (0x) edge (11);
      \draw[rf,bend left=10] (0y) edge (12);
  \draw[mo,bend left=15] (0y) edge (21);
\end{tikzpicture}}}\]
\end{center}
\end{minipage}
\caption{Examples of programs with RMWs, see \cref{ex:rmw}}%
 \label{fig:rmw}
\end{figure}
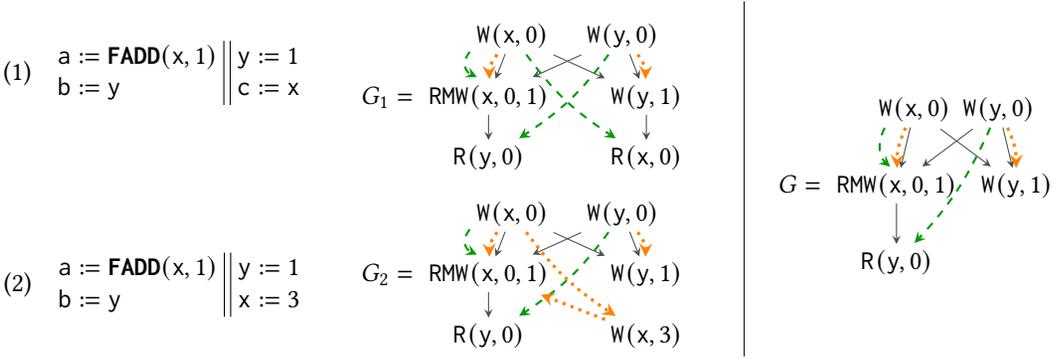

\begin{example}%
\label{ex:rmw}
Consider the tests
and corresponding candidate execution graphs in \cref{fig:rmw}.
Program (1) is not robust: $G_1$ is C11-consistent but \SC-inconsistent.
However, Program (2) is robust.
In particular, due to~\ref{ax:at}, $G_2$ is C11-inconsistent, and thus $G_2$ does not serve as a counterexample to robustness.
Note, however, that if we replace the $\faddInst{\creg1}{\cloc1}{1}$ instruction in $\ctid1$ with a normal write, 
Program (2) also becomes non-robust.
In this modified case, $G_2$ with $\evlab{\lU}{}{\cloc1}{0,1}$ replaced by the label of the new write
(and the $\lRF$-incoming edge toward this node removed) would become a counterexample to robustness.

The graph $G$ on the right serves as a witness of non-robustness for Program (1). 
If we execute $\ctid1$ first and then $\ctid2$, we reach a scenario where $\ctid1$ has the timestamp of the RMW 
in its $\vcT_\SC$ LC but not in its $\vcT_\HB$ LC\@.
Therefore, to maintain soundness and avoid reporting a spurious violation for Program (2), 
we adjust the condition being checked: 
when $\ctid2$ is about to perform a write, 
we compare the RMW LCs $\vcTU_\SC$ to $\vcTU_\HB$, rather than $\vcT_\SC$ to $\vcT_\HB$.
The modified LCs do not count RMWs in their timestamps and, 
for the run discussed, would carry only the initial write in both clocks.
\end{example}

\subsection{Supporting \RLX}%
\label{sec:rc20}

Our approach extends to a substantial fragment of the C11 model that includes, 
in addition to the release/acquire accesses discussed so far, 
non-atomic accesses, atomic relaxed accesses, and release/acquire fences. 
Their semantics follows \RLX as defined in~\cite{popl21}. 
In particular, the ``out-of-thin-air'' problem is conservatively avoided by disallowing $G.\lPO \cup G.\lRF$ cycles
like in RC11~\cite{scfix}.
\SC-fences are modeled using release/acquire fences and RMWs to a dedicated, otherwise unused location.
\SC accesses are not included in this model.

The inclusion of relaxed accesses and release/acquire fences complicates synchronization. 
Instead of incorporating $\lRF$ directly into $\lHB$, 
\RLX employs a derived relation, $\lSW$ (\emph{synchronized with}), 
which includes selected $\lRF$ edges along with other synchronization patterns 
via release/acquire fences and ``release sequence'' 
formed by $\lRF$ edges linked through intermediary RMWs. 

To support this extension, we refine the tracking of the $\lHB$ relation. 
Rather than the single $\vcT_\HB$ LC, we now maintain three 
LCs---$\vTtp{\relv}$, $\vTtp{\curv}$, and $\vTtp{\acqv}$---that 
capture the ``release'', ``current'', and ``acquire'' views of each thread, 
with similar adaptations for the RMW LCs discussed in \cref{sec:rmws}.
(Notably, the \SC LCs that track $\lHBSC$ remain unchanged.) 
This approach is inspired by the operational modeling of $\lHB$ in the Promising Semantics~\cite{promise}. 

Non-atomics also require attention.
Since data races on non-atomics are considered program errors, 
and since we assume the robustness verification is always complemented by dynamic race detection, 
the robustness analysis simply ignores non-atomic accesses.
But, to faithfully detect data races on non-atomics, 
we do need to properly track \RLX's happens-before relation,
$\lHB$, with the different synchronization patterns mentioned above.
Indeed, a data race on non-atomics means that the conflicting accesses are not related by $\lHB$,
and 
while robustness of the atomic accesses in the program 
allows us to ignore weak memory behaviors for atomics,
the whether or not two conflicting non-atomic accesses are synchronized by $\lHB$
does depend on the way atomics and fences are used (even in \SC-consistent graphs).
For that matter, we devised a race detection algorithm (given in 
\iflong
\cref{app:race}),
\else
\cite[§B]{appendix}),
\fi
based on standard vector clocks, which, in particular, fully supports release/acquire fences.
We note that robustness of atomics is essential for the completeness of this algorithm; 
otherwise, it is affected by the ``observer effect'' (see \cref{ex:obs}).

\subsection{Avoiding Benign Robustness Violations}%
\label{sec:wait}

In some situations, achieving SC is overly costly, 
and certain robustness violations can be tolerated. 
A common example involves a thread in a busy-wait loop, 
monitoring a specific value---such as a flag set by another thread---before proceeding. 
Under SC, the loop would terminate when the required value becomes visible in memory. 
Under weakly consistent memory, 
the thread may observe a stale value and continue to spin in the loop until the update propagates. 
A similar scenario is a CAS loop, 
where a thread repeatedly attempts to atomically update a variable until it succeeds. 
In weakly consistent models, the CAS may fail and retry even when it would succeed under SC\@.
Such robustness violations are typically considered benign as they do not compromise safety.

To support these benign robustness violations, 
which were also handled by previous static analysis~\cite{pldi19}, 
we require additional annotations in the input program. 
(Static analysis could automate this translation, but it is beyond our current scope.)
Busy-wait loops should be explicitly marked. 
We use two language constructs for this purpose: 
a $\waitInst{\loc}{\val}$ instruction, which blocks execution until it can read $\val$ from $\loc$, 
and a $\bcasInstn(\loc,\val_\lR,\val_\lW)$ instruction, 
which blocks until a CAS operation on $\loc$ succeeds in changing $\val_\lR$ to $\val_\lW$.
These instructions are designed as blocking, 
ensuring that execution graphs generated from programs using these constructs 
do not include events for failed attempts. 
This approach allows us to retain the robustness definition as it is: 
every \RLX-consistent execution graph of the program must also be \SC-consistent.

From the annotated code, 
we compile a finite set of ``critical location-value pairs'',
$\Crit \suq \Loc \times \Val$,
and track additional information associated for each pair.
Specifically, for each $\tup{\loc,\val} \in \Crit$ and thread $\tid \in \Tid$, 
we maintain a timestamp $\vcTV(\tid)(\loc,\val)$ that 
records the largest timestamp of a write $w_\loc^\val$ of $\val$ to $\loc$ 
that occurs $\lMO$-before a write $w$ that has $\lHBSC$ to thread $\tid$.
When thread $\tid$ is about to execute $\waitInst{\loc}{\val}$, 
we compare $\vcT_\HB(\tid)(\loc)$ to $\vcTV(\tid)(\loc,\val)$,
and flag a violation if $\vcT_\HB(\tid)(\loc) \leq \vcTV(\tid)(\loc,\val)$. 
This check ensures that not only $\vcT_\HB(\tid)(\loc) < \vcT_\SC(\tid)(\loc)$---permitting 
the thread to read from a stale write under $\RLX$ as before---but 
also that the stale write indeed holds the required value $\val$.
Additionally, we maintain timestamps $\vcWV(\loca)(\loc,\val)$ and $\vcMV(\loca)(\loc,\val)$ 
to be able to update $\vcTV(\tid)(\loc,\val)$ consistently.

\bigskip
\begin{example}
Consider the following~\ref{prog:bar} program implementing a simple barrier:

\smallskip
\noindent
\begin{minipage}{0.5\textwidth}
\begin{equation}
\tag{BAR}\label{prog:bar}
\inarrII{\assignInst{\cloc1}{1} \\ \while{\cloc2 \neq 1}{\skipc}}
{\assignInst{\cloc2}{1} \\ \while{\cloc1 \neq 1}{\skipc}}
\end{equation}
\end{minipage}\hfill
\begin{minipage}{0.45\textwidth}
\[G=\vcenter{\hbox{\begin{tikzpicture}[yscale=0.8,xscale=0.8]
  \node (0x)  at (-0.7,2) {$\evlab{\lW}{}{\cloc1}{0}$};
  \node (0y)  at (0.7,2) {$\evlab{\lW}{}{\cloc2}{0}$};
  \node (11)  at (-1,1) {$\evlab{\lW}{}{\cloc1}{1}$ };
  \node (12)  at (-1,0) {$\evlab{\lR}{}{\cloc2}{0}$ };
  \node (21)  at (1,1) {$\evlab{\lW}{}{\cloc2}{1}$ };
  \draw[po] (11) edge (12);
  \draw[po] (0x) edge (11) edge (21);
  \draw[po] (0y) edge (11) edge (21);
    \draw[mo,bend right=15] (0x) edge (11);
      \draw[rf,bend left=10] (0y) edge (12);
  \draw[mo,bend left=15] (0y) edge (21);
\end{tikzpicture}}}\]
\end{minipage}

\noindent
Program~\ref{prog:bar} is not robust. 
Under weak memory semantics, 
both while loops might repeatedly read $0$ for a while, 
whereas under SC, at least one of them must read $1$ in its very first iteration. 
However, this difference does not impact the barrier's safety, 
as each thread can only proceed after the other has raised its flag.
The execution graph $G$ on the right serves as a non-robustness witness for~\ref{prog:bar}. 
In a run where $\ctid1$ executes a write and a read before $\ctid2$ starts, 
a robustness violation is detected just before $\ctid2$ reads $\cloc1$, as we observe:
$0 = \vcT_\HB(\ctid2)(\cloc1) < \vcT_\SC(\ctid2)(\cloc1) = 1$.

To mask this robustness violation, 
we expect the user to provide Program~\ref{prog:barw11} below, 
which explicitly captures the intended behavior of the busy loops. 

\noindent
\begin{minipage}{0.32\textwidth}
\begin{equation}
\tag{BARW\ensuremath{_{1,1}}}\label{prog:barw11}
\inarrII{\assignInst{\cloc1}{1} \\ \waitInst{\cloc2}{1} }
{\assignInst{\cloc2}{1}\\\waitInst{\cloc1}{1} }
\end{equation}
\end{minipage}\hfill
\begin{minipage}{0.32\textwidth}
\begin{equation}
\tag{BARW\ensuremath{_{0,2}}}\label{prog:barw02}
\inarrII{\assignInst{\cloc1}{1} \\ \waitInst{\cloc2}{0} }
{\assignInst{\cloc2}{1}\\\waitInst{\cloc1}{2} }
\end{equation}
\end{minipage}\hfill
\begin{minipage}{0.32\textwidth}
\begin{equation}
\tag{BARW\ensuremath{_{0,0}}}\label{prog:barw00}
\inarrII{\assignInst{\cloc1}{1} \\ \waitInst{\cloc2}{0} }
{\assignInst{\cloc2}{1}\\\waitInst{\cloc1}{0} }
\end{equation}
\end{minipage}

\smallskip
\noindent
The graph $G$ above is not a non-robustness witness for~\ref{prog:barw11}, 
as it is not even a candidate execution graph of this program.
However,~\ref{prog:barw11} does not fully illustrate the issue. 
For example, $G$ is still an execution graph of the \emph{robust} program~\ref{prog:barw02} above. 
In~\ref{prog:barw02}, under the described run, we observe:
$0 = \vcT_\HB(\ctid2)(\cloc1) < \vcT_\SC(\ctid2)(\cloc1) = 1$.
The specialized tracking of the pair $\tup{\cloc1,2}$ successfully prevents 
the reporting of such spurious violations: 
here, we have $\vcTV(\ctid2)(\cloc1,2) = -1$ 
(the initial timestamp we use when no write of $2$ to $\cloc1$ exists), 
and $\vcT_\HB(\ctid2)(\cloc1) > \vcTV(\ctid2)(\cloc1,2)$.

We also note that a naive approach that checks for robustness violations 
only when $\waitInstn$ can succeed is flawed. 
To illustrate, consider the program~\ref{prog:barw00} above. 
It is not robust because only weak memory semantics allow \emph{both} $\waitInstn$ to proceed. 
Since the dynamic robustness analysis runs the instrumented program under \SC, 
we would never encounter a scenario where both $\waitInstn$ instructions succeed, 
thus this naive solution would miss the violation entirely.
\end{example}

\section{Implementation}%
\label{sec:impl}

We have implemented our approach in an open-source tool called
\toolname (\toolnamefull). In this section, we discuss some design
choices and optimizations in \toolname's implementation.

A dynamic analysis is performed either online (``on-the-fly''), by instrumenting
the program for monitoring its executions, or offline (``post-mortem'') by first
obtaining a program trace, and then analyzing it.  The online approach
is followed by tools like \tsan,
while the offline approach is typically
followed by experimental predictive race detectors~\cite{Kini17,lics0001P020,Pavlogiannis20}.
We opted for an online analysis in order to quickly and accurately
provide users with debugging information (\eg source-code-level names, stack trace).
Although it is possible to store all relevant debugging information for
post-mortem analysis, we found that doing so makes parsing and analyzing the trace much slower.

We implemented \toolname on top of the \tsan framework~\cite{www:tsan}.
\tsan is a dynamic race detector that ships by default with clang, and can handle C/C++ code.
Building on top of \tsan's infrastructure allows us to scale to large, real-world
programs that use various features of C/C++, including,
\eg~thread local storage, pointers, and dynamic thread creation.

To be able to dynamically monitor a program, \tsan (and, by
extension, \toolname) uses a non-invasive instrumentation, meaning that user
programs are compiled with a special flag, and there is no need to include
custom headers or link against heavyweight libraries.
Upon compilation of a program to LLVM-IR, a compiler pass transforms
all memory accesses to special function calls, which are used
to maintain the required instrumentation.
\toolname only tracks atomic accesses and fences.
For value-tracking (see \cref{sec:wait}), we require users to
call \toolname-specific primitives corresponding to the
$\waitInstn$, $\bcasInstn$, and strong $\casInstn$ instructions.
A future static analysis could automate the insertion of these calls in many cases.

\toolname extends the existing infrastructure of \tsan to maintain the
LCs described in \cref{sec:algorithm} and \cref{sec:extensions}
(see 
\iflong
\cref{app:rlx} 
\else
\cite[§A]{appendix}
\fi
for the full algorithm).
The LC maintenance also ensures \SC semantics,
since every two accesses to the same location are synchronized via locks.
This enables us to formally apply \cref{thm:main}.
To detect data races on non-atomics, we also modify \tsan
to execute our race-detection algorithm discussed in \cref{sec:rc20}
(given in 
\iflong
\cref{app:race}).
\else
\cite[§B]{appendix}).
\fi

To obtain a scalable implementation, our optimization efforts focused on
efficiently implementing LC operations, as real-world programs comprise
thousands of locations, and LCs (indexed by locations) are more difficult to
handle than vector clocks (indexed by thread identifiers).
We implement LCs as sorted
arrays of pairs.  When looking for a
specific address in an LC, we perform a brute-force search if the LC
size is small (this yields better cache locality and branch
prediction), and fall back to a binary search for larger LCs.
To efficiently implement LC merge, we avoid allocating a new LC for
the result. Instead, we first update the common addresses of the two
LCs in the LHS using a linear pass, we then add the extra RHS
addresses to the LHS, and finally re-sort the array.
We leave for future work further optimizations in this aspect.

Finally, to enhance user friendliness, we ensured \toolname error reports
can be used to establish robustness.
Whenever a robustness violation is found, \toolname reports the
offending events to the user. Concretely, an \toolname report includes the
write event whose timestamp does not satisfy the required condition
for robustness, and the concurrent memory access that triggers this
check. To assist locating the source of the violation, \toolname also
provides the corresponding stack traces and source-level instructions.

\section{Evaluation}%
\label{sec:eval}

In \toolname's evaluation,
we aim to demonstrate the following points:
\begin{itemize}
\item[\cref{sec:eval-static}] \toolname finds almost all robustness
  violations found by exhaustive (static) techniques, but faster and
  using less memory;
\item[\cref{sec:eval-overhead}] \toolname incurs an acceptable overhead over \tsan; and
\item[\cref{sec:eval-case-studies}] \toolname can detect and repair
  robustness violations in real-world codebases, often comprising
  thousands of lines of code, much beyond the state-of-the-art
  in robustness checking.
\end{itemize}

To do so, we ran \toolname on a representative set of both synthetic
and real-world benchmarks.
As our evaluation demonstrates, \toolname is an excellent alternative
to instrumentation-heavy or exhaustive techniques, as
\begin{enumerate*}[label=(\alph*)]
\item it consumes much less memory and scales to many more threads
  than exhaustive techniques,
\item it detects robustness violations in large codebases within seconds, and
\item requires minimal expertise to be run and repair such violations.
\end{enumerate*}

\paragraph{Experimental Setup}
All tests were done on a
Intel®Xeon® Gold 6132 CPU @2.60GHz GNU/Linux machine with 187GB of RAM\@.
We used a timeout of 60 minutes, and a memory limit of 2GB\@.
Memory consumption is reported in MB and time in seconds.
The values are averaged across 10 runs.

\subsection{Comparison with Exhaustive Techniques}%
\label{sec:eval-static}

We begin by comparing \toolname with \rocker, a static robustness
verifier from~\cite{popl21} following the approach described in \cref{sec:robustness_verification}.
\rocker accepts as input programs written in a toy language called TPL,
which are then instrumented, translated to Promela, and verified by the SPIN model checker~\cite{spin}.
Crucially, the input for \rocker has to be a finite-state program, so exhaustive search converges.

For this part of our evaluation, we restrict ourselves to synthetic
benchmarks since
\begin{enumerate*}[label=(\alph*)]
\item translating C/C++ to TPL is arduous, and
\item as our results below show, \rocker cannot handle larger programs
  anyway.
\end{enumerate*}
We use synthetic versions of well-known synchronization
algorithms obtained from previous work~\cite{popl21},
including Dekker's mutual exclusion algorithm~\cite{EWD123} 
and Peterson's algorithm~\cite{peterson}.

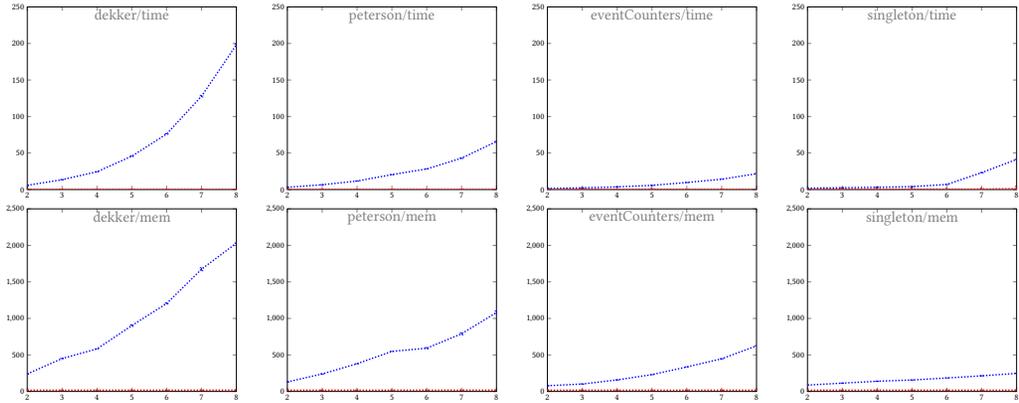
\begin{figure}{%
\newcommand{\timePlot}[1]{%
\begin{tikzpicture}[scale=0.6]
\begin{axis}[
    title style={at={(0.5,1)},anchor=north,yshift=-5,gray},title={\huge #1/time},
    xmin=1.5,xmax=9.5,
    ymin=0,ymax=250,
    scale only axis,
    ylabel shift=-12pt,
    xmin=2,xmax=8,
    legend pos=north west,
]

\addplot[blue, ultra thick,dotted, discard if not={name}{#1},
error bars/.cd, y dir = both, y explicit]
table[x=threads, y=rocker-time-mean,
  y error = rocker-time-std
    ] {csvRsan/NotRobustN.csv};
\addplot[red, ultra thick,dotted, discard if not={name}{#1},
error bars/.cd, y dir = both, y explicit]
table[x=threads, y=rsanTimeToErr-mean,
  y error = rsanTimeToErr-std,
    ] {csvRsan/NotRobustN.csv};
\end{axis}
\end{tikzpicture}}
\newcommand{\ramPlot}[1]{%
\hspace*{-4pt}%
\begin{tikzpicture}[scale=0.6]
\begin{axis}[
    title style={at={(0.5,1)},anchor=north,yshift=-5,gray},title={\huge #1/mem},
    scale only axis,
    xmin=2,xmax=8,
    ymin=0,ymax=2500,
    legend pos=north west,
]

\addplot[blue, ultra thick,dotted, discard if not={name}{#1}, error bars/.cd, y dir = both, y explicit]
table[x=threads, y=rocker-ram-mean,
  y error = rocker-ram-std,
    ] {csvRsan/NotRobustN.csv};
\addplot[red, ultra thick,dotted, discard if not={name}{#1}, error bars/.cd, y dir = both, y explicit]
table[x=threads, y=rsanSingleExecution-ram-mean,
  y error = rsanSingleExecution-ram-std,
    ] {csvRsan/NotRobustN.csv};
\end{axis}
\end{tikzpicture}}
\newcommand{\plotnonrobust}[1]{%
\scalebox{0.55}{\begin{minipage}{170pt}{
\timePlot{#1}\\\ramPlot{#1}
}\end{minipage}}
}
\plotnonrobust{dekker}
\plotnonrobust{peterson}
\plotnonrobust{eventCounters}
\plotnonrobust{singleton}

}
\caption{Time (s) and memory consumption (MB) for {\color{blue}{\rocker}} and {\color{red}{\toolname}} per
  number of threads}
\label{fig:eval-static:plot}
\end{figure}

We have compared \toolname and \rocker on both robust and non-robust
benchmarks.
First, let us consider the non-robust benchmarks (\cref{fig:eval-static:plot}).
For this comparison, we report the average time/memory that
\toolname and \rocker need to detect a robustness violation, respectively.
As can be seen, \toolname is able to find robustness violations much more quickly and
with much less memory than \rocker. Indeed, while \rocker often needs
minutes and gigabytes of memory the report a robustness violation,
\toolname typically reports violations within seconds and with constant memory consumption.
The number of runs required to find a robustness violation in these tests
ranges between a single run (\bmark{dekker} and \bmark{peterson})
to 30 (\bmark{eventCounters} for 8 threads). 

In addition to the non-robust benchmarks in \cref{fig:eval-static:plot}, we also
tested non-robust versions of Lamport's fast mutex algorithm~\cite{lamport-fast-mutex}
and of sequence lock~\cite{seqlock}, which 
starkly highlight the differences between the two tools.
On the one extreme, in the \bmark{lamport} benchmark, \toolname is able to immediately
find a robustness violation, while \rocker runs out of memory even for just three threads.
On the other extreme, however, in the \bmark{seqlock} benchmark, \toolname is unable to
find the robustness violation within the allocated time budget.
Indeed, in this benchmark the violation requires an intricate scheduling to
manifest, but the operating system does not context-switch in the appropriate
places. Although employing a more sophisticated scheduling algorithm in
\toolname would help in such cases, we leave such an implementation for
future work.

Overall---and despite the synthetic/non-representative nature of
the tests---we observe that \toolname scales to more threads, and
consumes much less memory than \rocker.
More importantly, we also observe that \toolname finds
robustness violations quickly, an observation also confirmed by our
realistic tests (see \cref{sec:eval-case-studies}), and which is
useful for determining whether a given program is robust.

Let us now consider some robust programs.
We constructed robust versions of the benchmarks in \cref{fig:eval-static:plot},
and measure how much memory the two tools used, as well as how many executions
\toolname managed to explore until \rocker established robustness.
The results followed trends similar to the non-robust case.
\toolname used a constant amount of RAM of ${\sim}18.5$MB while \rocker required from ${\sim}90$MB (for
small tests) up to ${\sim}2$GB, at which point it exceeded our memory limit.
In terms of time, \toolname needed ${\sim}0.03$s to execute one run of any benchmark,
which means it was able to run up to thousands of runs for most benchmark in the time
\rocker required to establish robustness (this varied from ${\sim}1$s to ${\sim}5$m).

We note that all benchmarks we used only access a few
memory locations. Since \rocker cannot handle even small synthetic
tests, it is unlikely to scale to realistic benchmarks with many
locations.

\subsection{Real-World Case Studies}%
\label{sec:eval-case-studies}

We used \toolname to establish robustness
of large, non-trivial concurrent programs.
Whenever \toolname found a violation, we added an SC fence
around the offending instructions or strengthened their access modes, 
and reran \toolname until no new violations
were reported within $100$ iterations.
Since \toolname typically found robustness violations after 1--2 iterations,
we considered such a limit to be appropriate.
In all tests, we were able to ran \toolname without modifying the
underlying code.

\paragraph{Chase-Lev Deque}

Chase-Lev deque~\cite{chase2005dynamic} is a popular wait-free work-stealing queue.
\citet{Le13} carefully studied its implementation under weak memory
(and it was also studied in the context of a Rust library\footnote{\url{https://github.com/jeehoonkang/crossbeam-rfcs/blob/deque-proof/text/2018-01-07-deque-proof.md} (accessed November 2024)}),
while \citet{popl21} proved an idealized implementation robust, by adding a few more
fences and strengthening some access modes.

For our first case study, we checked whether
\toolname could prove robustness of an implementation that cannot be verified
with state-of-the-art tools like \rocker.
To that end, we obtained a popular Chase-Lev implementation\footnote{\url{https://github.com/ConorWilliams/ConcurrentDeque/}
(accessed November 2024)} (143 stars on Github), and checked its robustness.
All in all, \toolname discovered 15 robustness violations,
which we fixed by inserting 11 fences and promoting 4 relaxed accesses to release/acquire.
Each violation was discovered with about a single iteration on
average, confirming our observation in \cref{sec:eval-static} that
\toolname identifies robustness violations quickly.
Interestingly, with this approach the obtained (minimal) set
of \SC fences to be added to
the code to establish robustness is different from prior work . %
While not as efficient, \toolname's fence set was determined with minimal understanding of the code.
Fixing the robustness violations caused a 1.13x slowdown
in \toolname compared to the initial non-robust case, which we consider fair
considering the number of fences inserted and the overhead they incur for maintaining the instrumentation.
Finally, we compare the cost of establishing robustness using \toolname to what one would get when naively strengthening
\emph{all} atomic accesses to \SC, which on certain architectures entail a full memory barrier.
For the case of Chase-Lev, the latter would lead to 28\% more barriers in the compiled code.

\paragraph{VSync Library}
\textsc{VSync}\footnote{\url{https://github.com/open-s4c/libvsync/}
  (accessed November 2024)}
is a concurrent library comprising
many different locking and concurrent-data-structure implementations, which were
verified under weak memory consistency using bounded model
checking~\cite{Oberhauser21}.
In this case study, we wanted to see whether these data structures are
also robust and, if not, what is the cost of making them robust.

We took all of \textsc{VSync}'s concurrent data
structures and checked their robustness.
\toolname discovered 0.75 robustness violations per
test on average, and each violation manifested within about 1.2
iterations. Fixing the violations required
38 extra SC fences in 15 files (data structures and their dependencies), which, depending on the application, may be
acceptable given the resulting simplicity in reasoning.
Similarly to our first case study above, fixing the robustness
violations incurred a 1.3x slowdown in \toolname compared
to the non-robust case.
In this case, however, establishing robustness by promoting all atomic
accesses to \SC would lead to 1,400\% more fences on average (26\% in the
best case).

\paragraph{Mimalloc Allocator}
\texttt{mimalloc}\footnote{\url{https://github.com/microsoft/mimalloc}
  (accessed November 2024)} is a memory allocator by
Microsoft that heavily uses weak memory atomics to achieve
good performance.
  \texttt{mimalloc} has over 10.6k stars on Github
and comprises about 12KLoc lines of C code, rendering it our largest case
study.
We used \texttt{mimalloc}'s standard stress testing infrastructure for
our client code, which uses 25 threads with a 25\%
load-per-thread workload, and creates over 60k distinct atomic locations,
leading to very large LCs, and thus increased monitoring cost.

\toolname identified 20 distinct robustness violations which we repaired
one by one. Even though each robustness violation was identified using
a single iteration, the time these iterations required varied.
The reason for the varying time is the fact that we aborted when
encountering robustness violations (thereby not running the whole
program), as well as the overhead induced by the added fences.
Achieving robustness by promoting all accesses to \SC would lead to
300\% more fences.

\subsection{Overhead over \tsan}%
\label{sec:eval-overhead}

\newcommand{\TSANmaxoverhead}{222.08}
\newcommand{\TSANnumtests}{154\xspace}
\newcommand{\TSANgeomeanoverhead}{1.52}%

Finally, we measure the overhead that \toolname induces over \tsan
(using \tsan as is, without its adaptions mentioned in \cref{sec:impl}),
which reports 2-20x overhead in execution time
on top of non-instrumented execution for typical programs~\cite{www:tsan}.
To do so, we ran \toolname and \tsan on the tests of
\cref{sec:eval-static} and \cref{sec:eval-case-studies}, as well as on a large
number of automatically generated benchmarks containing a varying
number of (randomly generated) memory accesses (\TSANnumtests tests in total).
Although the latter set is not representative of real-world workloads, it
was designed to bring \toolname to its limits, so that we can measure its
overhead over \tsan in a ``worst-case'' scenario.

\smallskip
\begin{wrapfigure}[9]{r}{.3\textwidth}
\pgfplotsset{compat=newest} %
 \vspace*{-4ex}
{\footnotesize{\centering \textbf{Overhead:}\\}
\begin{tabular}{@{\hsep}c@{}}
geo-mean: \TSANgeomeanoverhead,  %
max:  \TSANmaxoverhead   \\
~
\end{tabular}}
\\[-2ex]\begin{tabular}{@{}l@{}}
\begin{tikzpicture}[scale=0.4]
\begin{loglogaxis}[
    scale only axis,
    xlabel={\large\tsan time (s)},
    ylabel={\large\toolname time (s)},
    ylabel shift=-12pt,
    xmin=0.01,xmax=100,
    ymin=0.01,ymax=100,
    legend pos=north west,
]

\addplot[ultra thin,domain=0.01:100]
    { x };
\addplot[blue, only marks, mark size=.7pt, error bars/.cd, y dir = both, y explicit, x dir = both, x explicit]
    table[x=tsan-avg, y=rsan-avg%
    ] {plot/nonrobust.dat};
\addplot[blue, only marks, mark size=.7pt, error bars/.cd, y dir = both, y explicit, x dir = both, x explicit]
    table[x=tsan-avg, y=rsan-avg%
    ] {plot/robust.dat};

\end{loglogaxis}
\end{tikzpicture}
\end{tabular} %
\label{fig:eval-overhead}
\end{wrapfigure}
Our results are summarized on the right; on average
\toolname incurs a \TSANgeomeanoverhead{x} overhead on top of \tsan.
This overhead mostly occurs in tests with large thread sizes,
and is attributed to the number of memory locations used
(more locations imply larger location clocks), and the number of \SC
fences present. Indeed, while \tsan simply executes fences
with no instrumentation,
\toolname has to instrument \SC
fences, and following the \RLX memory model,
it models them as RMWs to an otherwise unused location.
Thus, \SC fences in \toolname
contend for the same lock, which introduces
a performance penalty.

In certain extreme cases, this penalty might be substantial:
\eg for \texttt{mimalloc} (see \cref{sec:eval-case-studies}),
we observe that \toolname is \num[round-mode=places,round-precision=0]{\TSANmaxoverhead}~times slower than \tsan.
This is because \texttt{mimalloc} has a lot of fences (we used them to make
it robust), large threads, and a lot of memory locations:
it is a memory allocator, and its stress tests allocate a
lot of memory, thereby making \toolname's LCs explode.  Given that
\begin{enumerate*}[label=(\alph*)]
\item such cases were uncommon in our evaluation,
\item  the ``push-button'' nature
  of \toolname, and
\item  the fact that we were able to run \toolname on
real-world programs,
\end{enumerate*}
we consider \toolname's overall overhead acceptable, and leave
improvements to our treatment of fences and LCs as future work.

\section{Related and Future Work}%
\label{sec:related}

\paragraph{Robustness Verification}
The robustness of programs against relaxed memory models
presents a compelling correctness criterion,
enabling the reduction of reasoning and verification under complex,
unintuitive models to simpler reasoning that assumes sequential consistency (SC).
Robustness with respect to various hardware and programming language memory models
has been extensively studied in prior work,
such as~\cite{Alglave:2011, Derevenetc:2014, checkfence, Alglave:2017, Burckhardt:2008, Burnim2011,
tso-robustness,  Liu:2012, tso-robustness2, PSO:15, Abdulla:2015, Bouajjani18, Linden:2011,
Linden:2013, pldi19, popl21, Beillahi:2019, Beillahi:CAV19, Nagar24}.
Some of these studies present decision procedures,
often with PSPACE complexity, similar to the complexity of verifying finite-state concurrent programs under SC\@.
Others propose practical over-approximations of robustness,
offering guidance to programmers (and tools) on where to insert memory fences or strengthen access modes.
To our knowledge, all existing research has focused on the static verification of robustness,
and our paper is the first to propose a dynamic approach.

\paragraph{Enforcing Sequential Consistency}
Since weak memory models are inherently complex, 
researchers have investigated whether their performance benefits justify the added complexity,
and explored enforcing \SC at the compiler level to simplify software development~\cite{sc-mm,Liu17,Liu21}.
Our focus on robustness stems from a similar perspective. 
However, as shown in~\cref{sec:eval}, 
ensuring robustness in specific programs requires significantly fewer fences 
than enforcing \SC through general compiler mappings.
Evaluating the performance overhead of making programs robust is left for future work.

\paragraph{Dynamic Race Detection}
Dynamic verification of concurrent programs primarily focuses on race detection,
as data races are widely regarded as one of the most frequent sources of bugs in concurrent systems.
Several tools have been developed and employed for dynamic data race detection,
including
Eraser~\cite{Savage97}, Djit$^+$~\cite{MultiRace}, Helgrind$^+$~\cite{Helgrind},
ThreadSanitizer (\tsan)~\cite{Serebryany09,www:tsan}, and FastTrack~\cite{Flanagan09}
(see~\cite{Yu17} for a comparative study).
More theoretical work has focused on developing more efficient algorithms to predict additional races
without increasing the time and memory overheads~(see, \eg~\cite{Kini17,Mathur22}).
None of these dynamic methods supports C11 atomics.
Our work extends \tsan, which has only limited best-effort support for C11
and is impacted by the observer effect, as discussed in \cref{sec:intro}.

\paragraph{\ctsan}
\citet{Lidbury17} redesigned \tsan's race detection algorithm to adapt it for C11.
Their tool, \ctsan, detects races caused by weak memory behaviors.
However, this comes at the cost of significantly more complex and heavyweight instrumentation compared to other race detectors.
Specifically, to overcome the observer effect,
\ctsan employs software store buffers to track the stores to each memory location during execution.
To allow non-SC behaviors, when a read access occurs,
\ctsan randomly selects a write from the store buffers
to provide the value for the read, subject to constraints that ensure consistency.
This approach makes \ctsan resemble random testing rather than a pure dynamic analysis.
A further development incorporates controlled randomized scheduling into \ctsan~\cite{LidburyD19}.

Completeness is compromised in \ctsan.
First, the software store buffers have to be bounded in practice.
Second, to maintain manageable complexity 
\ctsan assumes a memory model stronger than C11
that precludes certain weak behaviors obtained by reordering of consecutive writes:

\begin{example}%
\label{example:tsan11}
Consider the following test (known as 2+2W):

\noindent
\begin{minipage}{0.2\textwidth}
\[
\inarrII{\assignInst{\cloc1}{1} \\ \assignInst{\cloc2}{2} \\ \assignInst{a}{\cloc2}}
{\assignInst{\cloc2}{1} \\ \assignInst{\cloc1}{2} \\ \assignInst{b}{\cloc1}}
\]
\end{minipage}\hfill
\begin{minipage}{0.75\textwidth}
When using atomic writes with weak memory orderings (relaxed or even release),
C11 permits the outcome $a = b = 1$.
This behavior is observable on ARM multiprocessors,
which may reorder writes to different addresses.
\end{minipage}

\noindent
However, \ctsan~\cite{Lidbury17} would not detect any race that depends on this weak behavior.
The tool's algorithm assumes that the ``modification order'' must align with the execution order
(\ie acyclicity of $\lPO \cup \lRF \cup \lMO$ rather than of $\lPO \cup \lRF$ as we assume in our work)---a 
condition that excludes the weak behavior in this example.
\end{example}

\paragraph{\ctester}
A more recent tool, \ctester by \citet{Luo21}, also implements random exploration for detecting data races.
\ctester's instrumentation more accurately models the C11 memory model,
particularly by allowing cycles in $\lPO \cup \lRF \cup \lMO$,
which means it correctly handles cases like the one in \cref{example:tsan11}.
However, to achieve this, \ctester must maintain detailed information about the execution graph at each step,
resulting in significant memory usage, which grows with the number of write events executed.
In contrast, the memory consumption of our approach (which targets a different property)
is practically independent of the length of the execution.

\begin{example}%
\label{example:C11Tester}
To measure \ctester's memory consumption, we ran \ctester on~\ref{prog:mp} from \cref{fig:intro_rob}
with $10^N$ writes to a single (irrelevant)
variable in the producer thread.  While the program without
instrumentation requires $16$KB (irrespective of $N$),
\ctester consumed $26$KB for $N=1$ and $6.3$GB for $N=7$,
demonstrating that it requires memory linear in the size of the
explored execution graph.
(In contrast, \toolname requires $19$MB regardless of $N$.)
\end{example}

\paragraph{Model Checking}
Stateless model checkers,
like Nidhugg~\cite{%
Abdulla19, Abdulla24},
RCMC~\cite{rcmc}, and GenMC~\cite{Kokologiannakis19},
are carefully designed to efficiently explore all possible executions of a (loop-free) program.
While these tools can be adapted for random exploration
(with GenMC already having a preliminary version for this purpose~\cite{Kokologiannakis24}),
the memory usage per execution remains proportional to the length of the execution.

\paragraph{Testing}
Probabilistic concurrency testing~\cite{BurckhardtKMN10},
a randomized testing algorithm that offers theoretical guarantees on the probability of detecting bugs,
was applied to the C11 memory model by \citet{Gao23}.
Broadly speaking, this approach involves identifying an appropriate notion of \emph{bug depth}
and developing an execution sampling algorithm capable of detecting all bugs up to a certain depth with some probability.\ \citet{Gao23} proposed using the number of reads-from edges from one thread to another
 as the bug depth notion for weak memory programs.
However, their sampling algorithm fails to capture executions with cycles in $\lPO \cup \lRF \cup \lMO$,
as the one in \cref{example:tsan11}.
Our current approach does not directly control the scheduler,
relying instead on the OS for all scheduling decisions.
A promising direction for future work is to add such control and incorporate probabilistic concurrency testing of robustness,
specifically targeting violations that arise only under particular interleavings
(such as those in \bmark{seqlock} discussed in \cref{sec:eval-static}).

\begin{acks}
We thank the anonymous reviewers for their valuable feedback.
This work was supported by 
the European Research Council (ERC) under the European Union's Horizon 2020
research and innovation programme (grant agreement no.~851811)
and the Israel Science Foundation (grant numbers 814/22 and 2117/23).
\end{acks}

\section*{Data-Availability Statement}

The artifact is available at~\url{https://doi.org/10.5281/zenodo.15002567}.

\bibliography{main}

\iflong
\clearpage
\appendix

\newcommand{\vcGT}{{}\hat{\vcT}}
\newcommand{\vcGW}{{}\hat{\vcW}}
\newcommand{\vcGM}{{}\hat{\vcM}}
\renewcommand{\faddInst}[3]{\faddInstn({#1}, {#2},{#3})}

\section{Full Location-Clock-Based Algorithm}%
\label{app:rlx}

We present the complete LC algorithm, incorporating all extensions discussed in \cref{sec:extensions}. 
For further details on \RLX, the underlying memory model assumed in this work, we refer the reader to~\cite{popl21}. 
We also adopt various notations introduced in~\cite{popl21}.

\paragraph{Timestamp Assignments}
Given an execution graph $G$, the modification order $G.\lMO$ induces
timestamp mappings $G.\lTS, G.\lTSU: G.\sW \to \Time$ defined by:
\begin{align*}
G.\lTS(w) & \defeq \size{\set{w'\in \sW \setminus \Init \st \tup{w',w}\in G.\lMO^?}} \\
G.\lTSU(w) & \defeq \size{\set{w'\in \sW \setminus (\sU \cup \Init) \st \tup{w',w}\in G.\lMO^?}}
\end{align*}
$G.\lTS(w)$ counts the number of write events, including RMWs, 
that are ordered before $w$ in $G.\lMO$, including $w$ itself, but excluding the initialization write
(we have $G.\lTS(w) = \size{\set{w'\in \sW \st \tup{w',w}\in G.\lMO}}$).
In turn, $G.\lTSU(w)$ is similar but is only counts \emph{non-RMW} write events
(we have $G.\lTSU(w) = \size{\set{w'\in \sW \setminus \sU \st \tup{w',w}\in G.\lMO^?}} - 1$).

\paragraph{\SC Location Clocks}
$\vcT_\SC$ and $\vcTU_\SC$ assign an LC to every thread,
and $\vcW_\SC$, $\vcM_\SC$, $\vcWU_\SC$, and $\vcMU_\SC$ assign an LC to every location.
Their meaning is as follows (where $\max\emptyset\defeq0$):
\begin{align*}
	\vcT_\SC(\tid)(\loc) & =\max\set{G.\lTS(w) \st w \in \dom{[\sW_\loc] \seq G.\lHBSC^? \seq [G.\sE^\tid]}}
\\
	\vcW_\SC(\loca)(\loc) & =\max\set{G.\lTS(w) \st w \in \dom{[\sW_\loc] \seq G.\lHBSC^? \seq [G.\wmax{\loca}]}}
\\
	\vcM_\SC(\loca)(\loc) & =\max\set{G.\lTS(w) \st w \in \dom{[\sW_\loc] \seq G.\lHBSC^? \seq [G.\sE_\loca]}}
\\
      \multispan2{$\vcTU_\SC,\vcWU_\SC,\vcMU_\SC$ are similar using $G.\lTSU$ instead of $G.\lTS$\hfil}
\end{align*}

Initially, all LCs are empty (initialized to $\vc_\Init$).
At each program step, these LCs are maintained as follows.
When thread $\tid$ performs an access of type $\typ \in \set{\lR,\lW,\lU}$ to an atomic location $\loc$,
we execute the following updates (using the helper procedure on the right), 
where $\vc_1 \sqassn \vc_2$ is a shorthand for  $\vc_1 \eqassn \vc_1 \sqcup \vc_2$:

\smallskip
\begin{minipage}[t]{.45\textwidth}
\begin{algorithmic}[1]
	\If {$\typ\in\set{\lW,\lU}$}
	\State $\vcM_\SC(\loc)(\loc) \eqassn \vcW_\SC(\loc)(\loc)+1$
	\EndIf
	\State \Call{UpdateLSC}{$\vcT_\SC, \vcW_\SC, \vcM_\SC$}
	\If {$\typ=\lW$}
	\State $\vcMU_\SC(\loc)(\loc) \eqassn \vcWU_\SC(\loc)(\loc)+1$
	\EndIf
	\State \Call{UpdateLSC}{$\vcTU_\SC, \vcWU_\SC, \vcMU_\SC$}
\end{algorithmic}
\end{minipage}
\hfill
\begin{minipage}[t]{.45\textwidth}
\begin{algorithmic}[1]
	\Procedure{UpdateLSC}{$\vcGT, \vcGW, \vcGM$}
	\If {$\typ\in\set{\lW,\lU}$}
	\State $\vcGT(\tid) \sqassn \vcGM(\loc)$
	\State $\vcGM(\tid) \eqassn \vcGT(\tid)$
	\State $\vcGW(\tid) \eqassn \vcGT(\tid)$
\EndIf
\If {$\typ=\lR$}
	\State $\vcGT(\tid) \sqassn \vcGW(\loc)$
	\State $\vcGM(\tid) \sqassn \vcGT(\tid)$
\EndIf
\EndProcedure
\end{algorithmic}
\end{minipage}

\paragraph{Overwritten-Value Location Clocks}
We assume a (finite) collection of ``critical'' location-value pairs,
$\Crit \suq \Loc \times \Val$,
and let
$\lLOC(\Crit) \defeq \set{\loc \st \exists \val\ldotp \tup{\loc,\val}\in\Crit}$
and $\Crit(\loc) \defeq  \set{\val \st \tup{\loc,\val}\in\Crit}$
for every $\loc\in\lLOC(\Crit)$.
We require that if $\val$ is the expected value of a $\waitInstn$, $\bcasInstn$, or strong $\casInstn$ 
instruction for location $\loc$ in the input program, then $\tup{\loc,\val}\in \Crit$.
Then, we maintain \emph{location-value clocks} (LVC) for the locations mentioned in $\Crit$.
LVCs are defined like LCs, but they carry a timestamp for every location-value pair in $\Crit$,
as well as a ``default'' timestamp for every location mentioned in $\Crit$.
They represent functions in 
$\Crit \cup \lLOC(\Crit) \to \Time \cup\set{-1}$
($-1$, which is not needed in LCs, is used as the initial value).

We maintain LVCs similar to the LCs above:
$\vcTV$ and $\vcTUV$ assign an LVC to every thread,
and $\vcWV$, $\vcMV$, $\vcWUV$, and $\vcMUV$ assign an LVC to every location.
Their formal meaning is as follows, 
where $\sW_{\loc,\val}$ denotes the set of all write/RMW events to location $\loc$ writing value $\val$,
$\sW_{\loc,\nin \valset}$ denotes the set of all write/RMW events to location $\loc$ writing value $\val \nin \valset$,
and $\max{}_{-1} T$ returns the maximum of $T$ or $-1$ when $T=\emptyset$:
\begin{align*}
	\vcTV (\tid)(\loc,\val) & =\max{}_{-1}\set{G.\lTS(w) \st w \in \dom{[\sW_{\loc,\val}] \seq G.\lMO \seq G.\lHBSC^? \seq [G.\lE^\tid]}}
\\        
	\vcTV (\tid)(\loc) & =\max{}_{-1}\set{G.\lTS(w) \st w \in \dom{[\sW_{\loc,\nin\Crit(\loc)}] \seq G.\lMO \seq G.\lHBSC^? \seq [G.\lE^\tid]}}
\\        
	\vcWV (\loca)(\loc,\val) & =\max{}_{-1}\set{G.\lTS(w) \st w \in \dom{[\sW_{\loc,\val}] \seq G.\lMO \seq G.\lHBSC^? \seq [G.\wmax{\loca}]}}
\\        
	\vcWV (\loca)(\loc) & =\max{}_{-1}\set{G.\lTS(w) \st w \in \dom{[\sW_{\loc,\nin\Crit(\loc)}] \seq G.\lMO \seq G.\lHBSC^? \seq [G.\wmax{\loca}]}}
\\        
	\vcMV (\loca)(\loc,\val) & =\max{}_{-1}\set{G.\lTS(w) \st w \in \dom{[\sW_{\loc,\val}] \seq G.\lMO \seq G.\lHBSC^? \seq [G.\lE_\loca]}}
\\
	\vcMV (\loca)(\loc) & =\max{}_{-1}\set{G.\lTS(w) \st w \in \dom{[\sW_{\loc,\nin\Crit(\loc)}] \seq G.\lMO \seq G.\lHBSC^? \seq [G.\lE_\loca]}}
\\
      \multispan2{$\vcTUV,\vcWUV,\vcMUV$ are similar using $G.\lTSU$ instead of $G.\lTS$\hfil}
\end{align*}

Initially, all timestamps are initialized to $-1$.
At each program step, these LCs are maintained as follows.
When thread $\tid$ performs an access of type $\typ \in \set{\lR,\lW,\lU}$ to an atomic location $\loc\in\lLOC(\Crit)$,
we execute the following updates, where $\val_\prev$ is the value of $\loc$ in memory 
before the action:

\begin{multicols}{2}
\begin{algorithmic}[1]
	\If {$\typ\in\set{\lW,\lU}$}
		\If {$\tup{\loc,\val_\prev}\in\Crit$}
			\State $\vcMV(\loc)(\loc,\val_\prev) \eqassn \vcW_\SC(\loc)(\loc)$
		\Else
			\State $\vcMV(\loc)(\loc) \eqassn \vcW_\SC(\loc)(\loc)$
		\EndIf
    \EndIf
	\State \Call{UpdateLSC}{$\vcTV, \vcWV, \vcMV$}
	\If {$\typ=\lW$}
		\If {$\tup{\loc,\val_\prev}\in\Crit$}
			\State $\vcMUV(\loc)(\loc,\val_\prev) \eqassn \vcWU_\SC(\loc)(\loc)$
		\Else
			\State $\vcMUV(\loc)(\loc) \eqassn \vcWU_\SC(\loc)(\loc)$
		\EndIf
	\EndIf
	\State \Call{UpdateLSC}{$\vcTUV, \vcWUV, \vcMUV$}
\end{algorithmic}
\end{multicols}

These updates use the same helper procedure as above.
We assume that they run before the ones updating 
the \SC LCs above 
(since we need the values of $\vcW_\SC(\loc)(\loc)$ and $\vcWU_\SC(\loc)(\loc)$ \emph{before} their update).

\paragraph{Happens-Before Location Clocks}
The instrumentation for $\lHB$ consists of 
$\vTtp{\curv}$, $\vTtp{\acqv}$, $\vTtp{\relv}$,
$\vTUt{\curv}$, $\vTUt{\acqv}$, and $\vTUt{\relv}$
that assign an LC to every thread,
and $\vWt$ and $\vWU$ that assign an LC to every location.
Their semantics is as follows:
\begin{align*}
\vTtp{\curv}(\tid)(\loc) & =
\max\set{G.\lTS(w) \st w \in \dom{ [\sW_\loc] \seq G.\lRF^?\seq G.\lHB^? \seq [G.\lE^\tid]} }
\\
\vTtp{\relv}(\tid)(\loc) & =
\max\set{G.\lTS(w) \st w \in \dom{ [\sW_\loc] \seq G.\lRF^?\seq G.\lHB^? \seq [G.\sE^\tid \cap \sF^{\sqsupseteq\rel}]}}
\\
\vTtp{\acqv}(\tid)(\loc) & =
\max\set{G.\lTS(w) \st w \in \dom{[\sW_\loc] \seq G.\lRF^?\seq G.\lHB^? \seq ([\sE^{\sqsupseteq\rel}]\seq([\sF];G.\lPO)^?\seq G.\lRF^+)^?\seq [G.\sE^\tid]}}
\\
\vWt(\loca)(\loc) & =
\max\set{G.\lTS(w) \st w \in \dom{[\sW_\loc] \seq G.\lRF^?\seq G.\lHB^? \seq ([\sE^{\sqsupseteq\rel}]\seq([\sF];G.\lPO)^?\seq G.\lRF^*)^?\seq [G.\wmax{\loca}]}}
\\
      \multispan2{$\vTUt{\curv},\vTUt{\relv},\vTUt{\acqv},\vWU$ are similar using $G.\lTSU$ instead of $G.\lTS$\hfil}
\end{align*}

Initially, all LCs are empty (initialized to $\vc_\Init$).
At each program step, these LCs are maintained as follows.
When thread $\tid$ performs an access of type $\typ \in \set{\lR,\lW,\lU}$ to an atomic location $\loc$
with access mode (\aka memory ordering) $\mod \in \set{\rlx,\acq,\rel,\acqrel}$,
we execute the following updates:

\begin{multicols}{2}
\begin{algorithmic}[1]
\If {$\typ\in\set{\lW,\lU}$}
	\State $\vWt(\loc)(\loc) \eqassn \vWt(\loc)(\loc)+1$
\EndIf
\If {$\typ=\lW$}
	\State $\vWUt(\loc)(\loc) \eqassn \vWUt(\loc)(\loc)+1$
\EndIf
\State $\vTtp{\curv}(\tid)(\loc) \eqassn \vWt(\loc)(\loc)$
\State $\vTUt{\curv}(\tid)(\loc) \eqassn \vWUt(\loc)(\loc)$
\State $\vTtp{\acqv}(\tid)(\loc) \eqassn \vWt(\loc)(\loc)$
\State $\vTUt{\acqv}(\tid)(\loc) \eqassn \vWUt(\loc)(\loc)$
\If {$\typ\in\set{\lR,\lU}$} 
	\State $\vTtp{\acqv}(\tid) \sqassn \vWt(\loc)$
	\State $\vTUt{\acqv}(\tid) \sqassn \vWUt(\loc)$
 	\If  {$\mod\sqsupseteq \acq$}
	\State $\vTtp{\curv}(\tid) \sqassn \vWt(\loc)$
	\State $\vTUt{\curv}(\tid) \sqassn \vWUt(\loc)$
 	\EndIf 	
\EndIf
\If {$\typ=\lW$} 
	\If {$\mod \sqsupseteq \rel$} 
		\State $\vWt(\loc)  \eqassn \vTt{\curv}(\tid) $
		\State $\vWUt(\loc) \eqassn \vTUt{\curv}(\tid)$
	\Else
		\State $\vWt(\loc)  \eqassn \vTt{\relv}(\tid) $
		\State $\vWUt(\loc) \eqassn \vTUt{\relv}(\tid)$
 	\EndIf
\EndIf
\If {$\typ=\lU$} 
	\If {$\mod \sqsupseteq \rel$} 
		\State $\vWt(\loc)  \sqassn \vTt{\curv}(\tid) $
		\State $\vWUt(\loc) \sqassn \vTUt{\curv}(\tid)$
	\Else
		\State $\vWt(\loc)  \sqassn \vTt{\relv}(\tid) $
		\State $\vWUt(\loc) \sqassn \vTUt{\relv}(\tid)$
 	\EndIf
\EndIf
\end{algorithmic}
\end{multicols}

When $\tid$ performs a fence with mode $\mod \in \set{\acq,\rel,\acqrel}$, we execute the following updates:

\begin{multicols}{2}
\begin{algorithmic}[1]
 	\If  {$\mod \sqsupseteq \acq$}
	\State $\vTtp{\curv}(\tid) \eqassn \vTtp{\acqv}(\tid)$
	\State $\vTUt{\curv}(\tid) \eqassn \vTUt{\acqv}(\tid)$
	\EndIf
 	\If  {$\mod \sqsupseteq \rel$}
	\State $\vTtp{\relv}(\tid) \eqassn \vTtp{\curv}(\tid)$
	\State $\vTUt{\relv}(\tid) \eqassn \vTUt{\curv}(\tid)$
	\EndIf
\end{algorithmic}
\end{multicols}

Finally, $\sco$-fences are modeled as sequences of instructions 
$$\fenceInst{\acq} ; \faddInst{\floc}{0}{\acqrel} ; \fenceInst{\rel}$$
where $\floc$ is a distinguished otherwise-unused location.

\paragraph{Robustness Check}

Before every instruction of thread $\tid$ that performs an access to an atomic location $\loc$, 
we do the check below and declare a robustness violation if the specified condition holds:

\begin{itemize}
	\item Read/Weak  $\casInstn$:
		${\vTtp{\curv}}(\tid)(\loc) < {\vcT_\SC}(\tid)(\loc)$
	\item Write/Fetch-and-add:
		${\vTUt{\curv}}(\tid)(\loc) < {\vcTU_\SC}(\tid)(\loc)$
	\item $\waitInstn$ for value $\val$:
		${\vTtp{\curv}}(\tid)(\loc) \leq {\vcTV}(\tid)(\loc,\val)$
	\item $\bcasInstn$ with expected value $\val$:
		${\vTUt{\curv}}(\tid)(\loc) \leq {\vcTUV}(\tid)(\loc,\val)$
	\item Strong  $\casInstn$ with expected value $\val$:
		\begin{itemize}
			\item ${\vTUt{\curv}}(\tid)(\loc) \leq {\vcTUV}(\tid)(\loc,\val)$; or
			\item ${\vTtp{\curv}}(\tid)(\loc) \leq {\vcTV}(\tid)(\loc,\vala)$ for some $\vala\neq \val$; or
			\item ${\vTtp{\curv}}(\tid)(\loc) \leq {\vcTV}(\tid)(\loc)$
		\end{itemize}
\end{itemize}

\section{Dynamic Race Detection for Non-Atomic Accesses}%
\label{app:race}

\newcommand{\naloc}{n}
\newcommand{\vcNAR}{\mathbb{NR}}
\newcommand{\vcNAW}{\mathbb{NW}}
\newcommand{\vcTR}{\vTtp{\relv}}
\newcommand{\vcTC}{\vTtp{\curv}}
\newcommand{\vcTA}{\vTtp{\acqv}}
\newcommand{\conditional}[3]{#1 \mathbin{?} #2 \mathbin{:} #3}
\renewcommand{\vc}{\mathbb{V}}

We describe the instrumentation needed to perform race detection
for non-atomic accesses under the \RLX memory model, 
assuming the analyzed program is robust.
Robustness allows us to consider only \SC executions of the program.
The race detection instrumentation has to track the $\lHB$ relation as defined in \RLX.

Below, we assume a set $\NALoc$ of non-atomic locations
(disjoint from the set $\Loc$ of atomic locations).
We use $\naloc$ to refer to a non-atomic location.

\paragraph{Vector Clocks}
A \emph{thread epoch} is a pair $\epoch{\tid}{\ts}$
where $\tid\in\Tid$ and $\ts\in\Time$, and a
\emph{vector clock} (or VC) $\vc$ is a set of thread epochs such that
each thread $\tid$ appears in at most one thread epoch in $\vc$.
A VC $\vc$ essentially represents a function from $\Tid$ to $\Time$,
assigning the (unique) timestamp $\ts$ such that $\epoch{\tid}{\ts}\in \vc$ for every $\tid\in\Loc$
that appears in $\vc$, and $0$ for every other $\tid\in\Tid$.
We often identify VCs with the functions they represent,
writing, \eg~$\vc(\tid)$ for the timestamp $\tid$ assigns to $\loc$.
The empty VC, representing the function $\lambda \tid \ldotp 0$, is denoted by $\vc_\Init$.
We also identify a thread epoch $\epoch{\tid}{\ts}$ with a singleton VC $\set{\epoch{\tid}{\ts}}$.
We use the following standard predicate and operation on VCs:
\[\vc_1 \sqsubseteq \vc_2 \defeq \forall \tid \ldotp \vc_1(\tid) \leq \vc_2(\tid)
\qquad\qquad
\vc_1 \sqcup \vc_2 \defeq \lambda \tid \ldotp \max \set{\vc_1(\tid), \vc_2(\tid)}\]

\paragraph{Instrumentation and Race Detection}
The instrumentation consists of the following:
$\vcTR$, $\vcTC$, and $\vcTA$ assigning a VC to every thread;
$\vWt$ assigning a VC to every atomic location;
and $\vcNAW$ and $\vcNAR$ assigning a VC to every non-atomic location.
Their meaning is as follows (where $G.\lE^{\tida,{\sqsupseteq\rel}} \defeq \set{e\in G.\lE \st \lTID(e)=\tida \land
\lMOD(e) \sqsupseteq\rel}$):\footnote{The ternary conditional $\conditional{b}{p}{q}$ evaluates to $p$
if $b$ holds and to $q$ if $b$ does not hold.}
\begin{align*}
\vcTR(\tid)(\tida) & =
\size{G.\lE^{\tida,{\sqsupseteq\rel}} \cap \dom{G.\lHB^?\seq [G.\lE^\tid \cap \sF^{\sqsupseteq\rel}]} }
\\
\vcTC(\tid)(\tida) & =
\size{G.\lE^{\tida, \sqsupseteq\rel} \cap \dom{G.\lHB^?\seq [G.\lE^\tid]} } + (\conditional{\tida = \tid}{1}{0})
\\
\vcTA(\tid)(\tida) & =
\size{G.\lE^{\tida, \sqsupseteq\rel} \cap \dom{G.\lHB^?\seq [\sE^{\sqsupseteq\rel}]\seq([\sF] \seq G.\lPO)^? \seq G.\lRF^+ \seq [G.\lE^\tid \cap \sR^{\not \sqsupseteq\acq}]} }
\\
\vWt(\loc)(\tida) & =
\size{G.\lE^{\tida, \sqsupseteq\rel} \cap \dom{G.\lHB^?\seq [\sE^{\sqsupseteq\rel}]\seq([\sF] \seq G.\lPO)^? \seq G.\lRF^* \seq [G.\wmax{\loc}]} }
\\
\vcNAW(\naloc)(\tida) & =
\size{G.\lE^{\tida, \sqsupseteq\rel} \cap \dom{G.\lPO\seq[\sW_\naloc]} } + ( \conditional{G.\sW_\naloc^\tida \neq \emptyset}{1}{0})
\\
\vcNAR(\naloc)(\tida) & =
\size{G.\lE^{\tida, \sqsupseteq\rel} \cap \dom{G.\lPO\seq[\lR_\naloc]} } + ( \conditional{G.\sR_\naloc^\tida  \neq \emptyset}{1}{0})
\end{align*}

Initially, all VCs are empty (initialized to $\vc_\Init$),
except for $\vcTC(\tid)$ that is set to $1$ if $\lambda \tida \ldotp \conditional{\tida = \tid}{1}{0}$ for every $\tid\in\Tid$.
At each program step, these VCs are maintained as follows.
When thread $\tid$ performs an access of type $\typ \in \set{\lR,\lW,\lU}$ to an atomic location $\loc$
with access mode $\mod \in \set{\rlx,\acq,\rel,\acqrel}$,
we execute the following updates:

\begin{multicols}{2}
\begin{algorithmic}[1]
	\If {$\typ\in\set{\lR,\lU}$}
		\If {$\mod \sqsupseteq \acq$}
			\State $\vcTC(\tid) \sqassn \vWt(\loc)$
		\Else
			\State $\vcTA(\tid) \sqassn \vWt(\loc)$
		\EndIf
	\EndIf
	\If {$\typ=\lW$}
		\If {$\mod \sqsupseteq \rel$}
			\State $\vWt(\loc) \eqassn \vcTC(\tid)$
		\Else
			\State $\vWt(\loc) \eqassn \vcTR(\tid)$
		\EndIf
	\EndIf
	\Statex{}
	\Statex{}
	\If {$\typ=\lU$}
		\If {$\mod \sqsupseteq \rel$}
			\State $\vWt(\loc) \sqassn \vcTC(\tid)$
		\Else
			\State $\vWt(\loc) \sqassn \vcTR(\tid)$
		\EndIf
	\EndIf
	\If {$\typ\in\set{\lW,\lU} \land \mod \sqsupseteq \rel$}
		\State $\vcTC(\tid)(\tid) \eqassn \vcTC(\tid)(\tid) + 1$
	\EndIf
\end{algorithmic}
\end{multicols}

When $\tid$ performs a fence with mode $\mod \in \set{\acq,\rel,\acqrel}$, we execute the following updates:

\begin{multicols}{2}
\begin{algorithmic}[1]
	\If {$\mod \sqsupseteq \acq$}
		\State $\vcTC(\tid) \sqassn \vcTA(\tid)$
		\Statex
	\EndIf
	\If {$\mod \sqsupseteq \rel$}
		\State $\vcTR(\tid) \eqassn \vcTC(\tid)$
		\State $\vcTC(\tid)(\tid) \eqassn \vcTC(\tid)(\tid) + 1$
	\EndIf
\end{algorithmic}
\end{multicols}

In addition, $\sco$-fences are modeled as sequences of instructions 
$$\fenceInst{\acq} ; \faddInst{\floc}{0}{\acqrel} ; \fenceInst{\rel}$$
where $\floc$ is a distinguished otherwise-unused location.

Finally, when thread $\tid$ accesses a non-atomic location $\naloc$ with access type $\typ\in \set{\lR,\lW}$ 
we run the following algorithm. If the assertion does not hold, we report a data race on $\naloc$.

\begin{multicols}{2}
\begin{algorithmic}[1]
	\State \Call{Assert}{$\vcNAW(\naloc) \sqsubseteq \vcTC(\tid)$}
	\If {$\typ = \lR$}
		\State $\vcNAR(\naloc)(\tid) \eqassn \vcTC(\tid)(\tid)$
	\Else
		\State \Call{Assert}{$\vcNAR(\naloc) \sqsubseteq \vcTC(\tid)$}
		\State $\vcNAW(\naloc)(\tid) \eqassn \vcTC(\tid)(\tid)$
	\EndIf
\end{algorithmic}
\end{multicols}

\paragraph{Correctness}
The correctness of the above algorithm follows from the following theorem:

\begin{theorem}[WR \& WW-races]
Let $G$ be an execution graph,
$\tida\in\Tid$, $\naloc \in \NALoc$,
and $e$ be a $G.\lPO$-maximal event.
Suppose that $e \in G.\sR_\naloc \cup G.\sW_\naloc$.
Define:
$E_1= G.\lE^{\tida, \sqsupseteq\rel} \cap \dom{G.\lPO\seq[\sW_\naloc]}$,
$E_2  =G.\lE^{\tida, \sqsupseteq\rel} \cap \dom{G.\lHB^?\seq [e]}$,
$\ts_1 = \size{E_1} + ( \conditional{G.\sW_\naloc^\tida \neq\emptyset}{1}{0})$,
and
$\ts_2 = \size{E_2} + (\conditional{\tida = \lTID(e)}{1}{0})$.
Then, the following hold:
\begin{description}
\item[Soundness:] If $e$ does not participate in a race in $G$,
then $\ts_1 \leq \ts_2$.
\item[Completeness:] If $e$ races with some $w^\tida\in G.\sW_\naloc^\tida$ in $G$,
then $\ts_2 < \ts_1$.
\end{description}
\end{theorem}
\begin{proof}
For soundness, notice first, if $\tida = \lTID(e)$, then $E_1 \suq E_2$, and thus $\ts_1 \leq \ts_2$.
Second, if $\ts_1 = 0$, then we are also done.
Otherwise, $G.\sW_\naloc^\tida \neq \emptyset$.
Let $w^\tida$ be a $G.\lPO$-maximal event in $G.\sW_\naloc^\tida$.
We claim that $w^\tida$ races with $e$ in $G$.
Suppose otherwise.
Then, $\tup{w^\tida,e}\in G.\lHB$.
It follows that $E_1 \suq E_2$
and, moreover, some event $e^\tida\in G.\lE^{\tida, \sqsupseteq\rel}$ 
has $\tup{e^\tida,e}\in G.\lHB$ and $\tup{w^\tida,e^\tida}\in G.\lPO$.
Then, $e^\tida\in E_2 \setminus E_1$,
and so $E_1 \subsetneq E_2$.
Hence, $\ts_1 = \size{E_1} + 1 \leq \size{E_2} = \ts_2$.

For completeness, we claim that $E_2 \suq E_1$.
Let $e^\tida \in E_2$. Then, $e^\tida \in G.\lE^{\tida, \sqsupseteq\rel} $
and $\tup{e^\tida,e}\in G.\lHB^?$.
Now, if  $\tup{w^\tida,e^\tida}\in G.\lPO$, then $\tup{w^\tida,e}\in G.\lHB$, 
which contradicts the fact that $e$ races with $w^\tida$ in $G$.
Hence, we have $\tup{e^\tida,w^\tida}\in G.\lPO$, and so $e^\tida\in E_1$.
It follows that $\size{E_2} \leq \size{E_1}$.
Moreover, $G.\sW_\naloc^\tida\neq \emptyset$ and $\tid\neq \tida$
(since $e$ races with $w^\tida$ in $G$).
Therefore, $\ts_2 = \size{E_2} <  \size{E_1} + 1 = \ts_1$.
\end{proof}

Similarly, the following holds:

\begin{theorem}[RW-races]
Let $G$ be an execution graph,
$\tida\in\Tid$, $\naloc \in \NALoc$,
and $e$ be a $G.\lPO$-maximal event.
Suppose that $e \in G.\sW_\naloc$.
Define:
$E_1= G.\lE^{\tida, \sqsupseteq\rel} \cap \dom{G.\lPO\seq[\sR_\naloc]}$,
$E_2  =G.\lE^{\tida, \sqsupseteq\rel} \cap \dom{G.\lHB^?\seq [e]}$,
$\ts_1 = \size{E_1} + ( \conditional{G.\sR_\naloc^\tida \neq \emptyset}{1}{0})$,
and
$\ts_2 = \size{E_2} + (\conditional{\tida = \lTID(e)}{1}{0})$.
Then, the following hold:
\begin{description}
\item[Soundness:] If $e$ does not participate in a race in $G$,
then $\ts_1 \leq \ts_2$.
\item[Completeness:] If $e$ races with some $r^\tida\in G.\sR_\naloc^\tida$ in $G$,
then $\ts_2 < \ts_1$.
\end{description}
\end{theorem}

\fi
\end{document}